\documentclass[article]{IEEEtran}

\IEEEoverridecommandlockouts
\usepackage{amsmath}
\usepackage{amssymb}
\usepackage{amsthm}
\usepackage{enumerate}

\usepackage{epsfig}
\usepackage{subfigure}
\usepackage{psfrag}
\usepackage{cite}
\usepackage{stfloats}
\usepackage[mathscr]{euscript}
\usepackage{hyperref}
\hypersetup{backref, colorlinks=true, linkcolor=blue}

\graphicspath{{figures/}}

\newcommand{\mats}{\ensuremath{\mathcal{S}}}

\newcommand{\mata}{\ensuremath{\mathcal{A}}}

\newcommand{\matb}{\ensuremath{\mathcal{B}}}

\newcommand{\matn}{\ensuremath{\mathcal{N}}}

\newcommand{\matg}{\ensuremath{\mathcal{G}}}

\newcommand{\matW}{\ensuremath{\mathcal{W}}}

\newcommand{\mreals}{\ensuremath{\mathbb{R}}}

\ifx\eqref\undefined
	\newcommand{\eqref}[1]{~(\ref{#1})}
\fi
\ifx\mod\undefined
	\def\mod{\mathop{\rm mod}}
\fi

\newcommand{\vect}[1]{{\bf #1}}

\newcommand{\bigo}[1]{O\left(#1\right)}

\newcommand{\Qinv}[1]{Q^{-1}\left(#1\right)}
\newcommand{\Var}[1]{\mathrm{Var}\left(#1\right)}

\newcommand{\E}[1]{\mathbb E\left[#1\right]}
\newcommand{\Prob}[1]{\mathbb P\left[#1\right]}
\newcommand{\1}[1]{1\left\{#1\right\}}

\def\argmin{\mathop{\rm argmin}}

\def\exp{\mathop{\rm exp}}

\def\EE{\mathbb{E}\,}

\def\PP{\mathbb{P}}

\def\eqdef{\stackrel{\triangle}{=}}

\def\sff{\mathsf f}
\def\sfd{\mathsf d}

\def\unifto{\mathop{{\mskip 3mu plus 2mu minus 1mu%
	\setbox0=\hbox{$\mathchar"3221$}%
	\raise.6ex\copy0\kern-\wd0%
	\lower0.5ex\hbox{$\mathchar"3221$}}\mskip 3mu plus 2mu minus 1mu}}

\ifx\lesssim\undefined
\def\simleq{{{\mskip 3mu plus 2mu minus 1mu%
	\setbox0=\hbox{$\mathchar"013C$}%
	\raise.2ex\copy0\kern-\wd0%
	\lower0.9ex\hbox{$\mathchar"0218$}}\mskip 3mu plus 2mu minus 1mu}}
\else
\def\simleq{\lesssim}
\fi

\ifx\gtrsim\undefined
\def\simgeq{{{\mskip 3mu plus 2mu minus 1mu%
	\setbox0=\hbox{$\mathchar"013E$}%
	\raise.2ex\copy0\kern-\wd0%
	\lower0.9ex\hbox{$\mathchar"0218$}}\mskip 3mu plus 2mu minus 1mu}}
\else
\def\simgeq{\gtrsim}
\fi

\newif\ifmapx
{\catcode`/=0 \catcode`\\=12/gdef/mkillslash\#1{#1}}
\edef\jobnametmp{\expandafter\string\csname jscc_fb_apx\endcsname}
\edef\jobnameapx{\expandafter\mkillslash\jobnametmp}
\edef\jobnameexpand{\jobname}
\ifx\jobnameexpand\jobnameapx
\mapxtrue
\else
\mapxfalse
\fi

\long\def\apxonly#1{\ifmapx{\color{blue}#1}\fi}

\newtheorem{thm}{Theorem}

\newtheorem{defn}{Definition}

\theoremstyle{remark}
\newtheorem{remark}{Remark}

\title{Joint source-channel coding with feedback}

\author{%
{\small\begin{minipage}{\linewidth}\begin{center}
\begin{tabular}{ccc}
{\large{Victoria Kostina}} 
& 
{\large{Yury Polyanskiy}}
&
{\large{Sergio Verd\'{u}}}\\
California Institute of Technology & 
Massachusetts Institute of Technology &
Princeton University \\
Pasadena, CA 91125 
& 
 Cambridge, MA 02139
&
Princeton, NJ 08544\\
\url{vkostina@caltech.edu} &
\url{yp@mit.edu}
& \url{verdu@princeton.edu}
\end{tabular}
\end{center}\end{minipage}}
\thanks{
This work was supported in part by the National Science Foundation (NSF)
under Grant CCF-1016625, by the NSF CAREER award under Grant CCF-1253205, and by the Center for Science of Information
(CSoI), an NSF Science and Technology Center, under Grant CCF-0939370.
}
}

\IEEEoverridecommandlockouts

\begin{document}
\maketitle
\begin{abstract}
This paper quantifies the fundamental limits of variable-length transmission of a general (possibly analog) source over
a memoryless channel with noiseless feedback, under a distortion constraint. We consider excess distortion, average distortion and
guaranteed distortion ($d$-semifaithful codes). 
In contrast to the asymptotic fundamental limit, 
a general conclusion is that allowing variable-length codes and feedback leads to a sizable
improvement in the fundamental delay-distortion tradeoff. In addition, we investigate the minimum energy
required to reproduce $k$ source samples with a given fidelity after transmission over a memoryless Gaussian channel, and we show that
the required minimum energy is reduced with feedback and an average (rather than maximal) power constraint.
\end{abstract}


\section{Introduction}

A famous result due to Shannon \cite{shannon1956zeroerror} states that feedback cannot increase the capacity of
memoryless channels. Burnashev \cite{burnashev1976data} showed that feedback improves the error exponent in a
variable-length setting. Polyanskiy et al. \cite{polyanskiy2011feedback} showed that allowing variable-length coding and
non-vanishing error-probability $\epsilon$ boosts the $\epsilon$-capacity of the discrete memoryless channel (DMC) by a
factor of $1 - \epsilon$. Furthermore, as shown in \cite{polyanskiy2011feedback}, if both feedback and variable-length coding
are allowed, then the asymptotic limit $\frac C {1 - \epsilon}$ is approached at a fast speed $\bigo{ \frac{\log \ell}{
\ell}}$ as the average allowable delay $\ell$ increases:\footnote{Unless explicitly noted, all $\log$ and $\exp$ in this paper are to arbitrary
matching base, which also defines units of entropy, information density and mutual information.}
\begin{equation}
(1 - \epsilon)\log M^\star(\ell, \epsilon) = \ell C + \bigo{ \log \ell } \label{intro:ch},
\end{equation}
where $M^\star(\ell, \epsilon)$ is the maximum number of messages that can be distinguished with average error probability
$\epsilon$ at average delay $\ell$, and $C$ is the channel capacity.  This is in contrast to channel coding at fixed
blocklength $n$ where, in most cases, the optimum speed of convergence of the maximal rate to capacity is $\bigo{\frac 1 {\sqrt n}}$, even when feedback is available, see \cite{polyanskiy2011feedback,fong2015AWGNfeedback}. Thus,
variable-length coding with feedback (VLF) not only boosts the $\epsilon$-capacity of the channel, but also markedly
accelerates the speed of approach to it. Moreover, zero-error communication is possible at an average rate arbitrarily
close to capacity via variable-length coding with feedback and termination (VLFT) codes, a class of codes that employs a special termination symbol to signal the end of
transmission, which is always recognized error-free by the receiver \cite{polyanskiy2011feedback}. As discussed in
\cite{polyanskiy2011feedback}, the availability of zero-error termination symbols models that common situation in which timing
information is managed by a higher layer whose reliability is much higher than that of the payload. 

In \cite{kostina2015varrate}, we treated variable-length data compression with nonzero excess distortion probability. In particular, we showed that in fixed-to-variable-length compression of a block of $k$ i.i.d. source outcomes, the minimum average encoded length $\ell^\star(k, d, \epsilon)$ compatible with probability $\epsilon$ of exceeding distortion threshold $d$ satisfies, under regularity assumptions, 
\begin{equation}
\ell^\star(k, d, \epsilon) = (1 - \epsilon) k R(d) -   \sqrt{\frac{k \mathcal V(d)}{2 \pi} } \mathrm{e}^{- \frac 1 2 \Qinv{\epsilon}^2 } + \bigo{\log k},  \label{intro:s}
\end{equation}
where $R(d)$ and $\mathcal V(d)$ are the rate-distortion and the rate-dispersion functions, and $Q$ is the standard normal complementary cumulative distribution function.  
The second term in the expansion~\eqref{intro:s} becomes more natural if one notices that for $Z \sim \mathcal N(0, 1)$, 
\begin{equation}
 \EE[Z  \cdot 1\left\{Z > \Qinv{\epsilon} \right\}] = {1\over \sqrt{2\pi}} \mathrm{e}^{-\frac 1 2 \Qinv{\epsilon}^2}.
\end{equation}
As elaborated in \cite{kostina2015varrate},
the expansion~\eqref{intro:s} has an  unusual feature: the asymptotic fundamental limit is approached from the ``wrong'' side, e.g. larger dispersions and shorter blocklengths reduce the average compression rate.

In this paper, we consider variable-length transmission of a general (possibly analog) source over memoryless channels with feedback,
under a distortion constraint. This variable-length joint source-channel coding (JSCC) setting can be viewed as a
generalization of the setups in \cite{polyanskiy2011feedback,kostina2015varrate}, which, as explained above, analyze channel coding and source coding, respectively. Related work includes an assessment of the
nonasymptotic fundamental limits of fixed-length JSCC in \cite{kostina2012jscc,campo2011random,wang2011dispersion},  a dynamic programming formulation of zero-delay JSCC with feedback in \cite{javidi2013dynamicJSCC}, and a practical variable-length
almost lossless joint compression/transmission scheme in \cite{caire2004almost}. Various feedback coding strategies are discussed in \cite{ovseevich1968capacity,fang1970optimum,chande1998joint,kafedziski1998joint,lu1999progressive,gastpar2003source,manakkal2005source} while pragramic feedback communication schemes 
that implement VLF include \cite{horstein1963sequential,ooi1998fast,ahlswede1971constructive,shayevitz2011optimal,schalkwijk1967transmission}.

We treat several scenarios that differ in how the distortion is evaluated and whether a termination symbol is allowed. 
In all cases, we analyze the {\it average delay} required to achieve the objective. The results in Section~\ref{sec:main} are summarized as follows:
\begin{itemize}
\item Under the {\it average distortion} criterion, $\EE[\mathsf d(S^k, \hat S^k)] \le d$, the minimal average delay $\ell^\star(k, d)$ attainable by VLF codes
transmitting $k$ source symbols satisfies
\begin{equation}
\ell^\star(k, d)\, C  =  k R(d)  + \bigo{\log k}. \label{intro:f2vAv}
\end{equation}
\item Under the {\it excess distortion probability} criterion, $\PP[\mathsf d(S^k, \hat S^k)>d] \le \epsilon$, the minimal average delay attainable by VLF codes
transmitting $k$ source symbols satisfies
\begin{align}
\ell^\star(k, d, \epsilon)\, C  &=  (1 - \epsilon) k R(d)  -  \sqrt{\frac{k \mathcal V(d)}{2 \pi} } \mathrm{e}^{- \frac 1 2 \Qinv{\epsilon}^2 }  \notag\\
&+ \bigo{\log k} \label{intro:f2v}. 
\end{align}
\item Under the {\it guaranteed maximal distortion} criterion, $\PP[\mathsf d(S^k, \hat S^k)>d] = 0$, the minimal average delay attainable by
VLFT codes transmitting $k$ source symbols satisfies
\begin{equation}
\ell^\star_t(k, d, 0)\, C = k R(d) + \bigo{\log k} \label{intro:f2v0}. 
\end{equation}
\end{itemize}

Similar to \eqref{intro:ch}, approaching the limits in \eqref{intro:f2vAv}, \eqref{intro:f2v} and \eqref{intro:f2v0} only requires an extremely thin
feedback link, namely, the decoder sends just a single acknowledgement signal (stop-feedback) once it is ready to decode.
\footnote{Stop-feedback is not to be confused with the termination symbol, which is a special symbol that the
\emph{encoder} can transmit error-free in order to inform the decoder that the transmission has ended and it is time to decode.
} Note that~\eqref{intro:f2v} exhibits significant similarities with \eqref{intro:s}: the asymptotic limit is approached from below, i.e. in contrast to the results in
\cite{polyanskiy2010channel,kostina2011fixed,kostina2012jscc}, smaller blocklengths and larger source dispersions are
beneficial. Note also that the leading term of the expansion in \eqref{intro:f2v} can be attained with variable-length codes without feedback. 

Interestingly, naive separate source/channel coding fails to attain any of the limits mentioned. For example, it approaches the asymptotic fundamental limit from above, e.g. even the
sign of the second term in \eqref{intro:f2v} is not attainable. This observation led us to believe, initially, that
competitive schemes in this setting should be of successive refinement and adaptation sort such as in~\cite{FV87,AS93},
or dynamic programming-like as in~\cite{GC11,javidi2013dynamicJSCC}. It turns out, however, that like the fixed-length JSCC achievability schemes in \cite{campo2011random,kostina2012jscc}, attaining
limits~\eqref{intro:f2vAv}-\eqref{intro:f2v0} requires a rather simple variation on the separation architecture: one
only needs to allow a variable-length interface between the source coder and the channel coder.  Note that typically,
separation is understood in the sense that the output of the compressor is treated as pure bits, which can
be arbitrarily permuted without affecting the performance of the concatenated scheme~\cite{hochwald1997tradeoff,wang2011dispersion}, provided that an inverse permutation is applied at the input of the decompressor.  Similarly, the performance of a {\it variable-length separated scheme} is insensitive to permutations (but not additions or deletions) of the bits at the output of the source coder. These semi-joint 
achievability schemes are the subject of Section~\ref{sec:newcodes}, and form the basis for the lossy joint source-channel codes, which are the subject of  Section~\ref{sec:main}.

Energy-limited codes are the subject of Section~\ref{sec:Ecodes}.  The optimal energy-distortion tradeoff achievable in the
transmission of an arbitrary source vector over the AWGN channel is studied in Section~\ref{sec:ed}. In that setting, disposing of the restriction on the number of channel uses per source
sample, we limit the total available transmitter energy $E$ and we study the tradeoff between the source dimension $k$,
the total energy $E$ and the fidelity of reproduction. Related prior work includes a study of asymptotic
energy-distortion tradeoffs \cite{jain2012energy} and a nonasymptotic analysis of the energy per bit required to reliably
send $k$ bits through an AWGN channel \cite{polyanskiy2011minimumenergy}. The main results in Section~\ref{sec:ed} are the following:
\begin{itemize}
\item Under the \textit{average distortion} constraint, the total minimum energy required
to transmit $k$ source samples over an AWGN channel \textit{with feedback} satisfies
\begin{equation}
 E^\star_f(k, d) \cdot {\log e\over N_0} = k R(d) + O(\log k)\,, \label{intro:f2vAv2}
\end{equation}
where $\frac {N_0}{2}$ is the noise power per degree of freedom.


\item Under the \textit{excess distortion probability} constraint, the minimum energy required
to transmit $k$ source samples over an AWGN channel \textit{without feedback} satisfies
\begin{align}
 &~E^\star \left( k, d, \epsilon \right) \cdot{\log e\over N_0}  \label{intro:E:2order}\\
 =&~ k R(d) + \sqrt{ k \left( 2 R(d) \log e + \mathcal V(d)\right) } \Qinv{\epsilon} + \bigo{ \log k }.
  \notag 
\end{align}
\item Under the \textit{excess distortion probability} constraint, the total minimum average energy\footnote{The energy constraint in \eqref{intro:E:2order} is understood on a per-codeword basis. The energy constraint in \eqref{intro:E:2orderf} is understood as averaged over the
 source and noise realizations.} required
to transmit $k$ source samples over an AWGN channel \textit{with feedback} satisfies
\begin{align}
&~ E^\star_f \left( k, d, \epsilon \right) \cdot {\log e \over N_0}  \label{intro:E:2orderf}\\ 
=&~ k R(d) (1 - \epsilon)  - \sqrt{\frac{k \mathcal V(d)}{2 \pi} } \mathrm{e}^{- \frac { (\Qinv{\epsilon})^2} 2 }   + \bigo{\log k}.
 \notag
\end{align}
Like \eqref{intro:f2v}, particularizing~\eqref{intro:E:2orderf} to $\epsilon=0$ also covers the case of guaranteed distortion. The leading term in the expansion \eqref{intro:E:2orderf} can be achieved even without feedback, as long as $\epsilon > 0$ and the power constraint is understood on the average over the codebook. 
\end{itemize}

We point out the following parallels between variable-length codes and energy-limited-codes. 
\begin{itemize}
 \item Under average distortion, in both cases the fundamental limit is approached at speed $\bigo{\frac{\log k}{k}}$ (cf. \eqref{intro:f2vAv}, \eqref{intro:f2vAv2}). 
 \item Allowing a non-vanishing excess-distortion probability and variable length (or variable power) boosts the asymptotic fundamental limit by a factor of $1 - \epsilon$. 
 \item Allowing both feedback and variable length (or variable power)  leads to expansions \eqref{intro:f2v}, \eqref{intro:E:2orderf}, in which shorter blocklengths are beneficial. 
 \item As long as feedback is available, in both variable length coding with termination and average energy-limited coding, guaranteed distortion ($\epsilon = 0$) can be attained, even though the channel is noisy.  
\end{itemize}

\section{Feedback codes for non-equiprobable messages}\label{sec:newcodes}

In this section we consider joint source-channel coding assessing reliability by the probability that the (possibly non-equiprobable) message is reproduced correctly. These results lay the foundation for the analysis of joint source-channel coding under distortion constraints presented in Section \ref{sec:main}.
Our key tools are two extensions of the channel coding bounds for the DMC with feedback
from~\cite{polyanskiy2011feedback}. VLF and VLFT codes are formally defined as follows.

\begin{defn}\label{def:vlf}
A variable-length feedback code (VLF) transmitting message $W$ (taking values in $\matW$) over the channel 
$\{P_{Y_i | X^i Y^{i-1}}\}_{i = 1}^\infty$ with input/output alphabets $\mata$/$\matb$ is defined by:
\begin{enumerate}
\item A random variable $U \in \mathcal U$ revealed to the encoder and decoder before the start of the transmission.
\item A sequence of encoding functions
$\mathsf f_n \colon \mathcal U \times \matW \times \mathcal B^{n - 1} \mapsto \mathcal A$, specifying the channel inputs
\begin{equation}
X_n = \mathsf f_n \left( U, W, Y^{n-1}\right) \label{f2v:Xn}.
\end{equation}
 \item A sequence of decoding functions $\mathsf g_n \colon \mathcal U \times \mathcal B^n \mapsto  {\mathcal W}$, $n = 1, 2, \ldots$
 \item A non-negative integer-valued random variable $\tau$, a stopping time of the filtration $\mathcal F_n =
 \sigma\left\{ U, Y_1, \ldots, Y_n\right\}$, which determines the time at which the decoder output is computed:
 \begin{equation}
\widehat W = \mathsf g_{\tau}(U, Y^\tau). 
\end{equation}
The average transmission duration is $\EE[\tau]$. 
\end{enumerate}
\end{defn}

A very similar concept is that of an VLFT code:
\begin{defn} \label{def:vlft}
A variable-length feedback code with termination (VLFT) transmitting $W \in \matW$ over the channel 
$\{P_{Y_i | X^i Y^{i-1}}\}_{i = 1}^\infty$ with input/output alphabets $\mata$/$\matb$ is defined 
similarly to VLF codes with the exception that condition 4) in Definition~\ref{def:vlf} is replaced by
\begin{enumerate} 
\item[4')] A non-negative integer-valued random variable $\tau$, a stopping time of the filtration 
$\matg_n=\sigma\{ W, U, Y_1, \ldots, Y_n\}$.
\end{enumerate}
\end{defn}

The idea of allowing the transmission duration $\tau$ to depend on the true message $W$ models the practical scenarios (see~\cite{polyanskiy2011feedback})  where there
is a highly reliable control layer operating in parallel with the data channel, which notifies the decoder when
it is time to make a decision.

Two special cases of Definitions \ref{def:vlf} and \ref{def:vlft} are stop-feedback and fixed-to-variable codes:
\begin{enumerate}
\item {\em stop-feedback codes} are a special case of VLF codes where the
encoder functions $\{f_n\}_{n=1}^\infty$ satisfy:
 \begin{equation}\label{eq:nofb}
         f_n(U, W, Y^{n-1}) = f_n(U, W)\,.
 \end{equation}
Such codes require very limited communication over feedback: only a
 single signal informing the encoder to stop transmitting symbols once the decoder is ready to decode.
\item {\em fixed-to-variable codes}, defined in~\cite{VS09}, are also required to
satisfy~\eqref{eq:nofb}, while the stopping time is\footnote{As explained
in~\cite{VS09}, this model encompasses fountain codes in which the decoder can get a highly
reliable estimate of $\tau$ autonomously without the need for a termination symbol.}
\begin{equation} \tau = \inf\{n\ge 1: g_n(U, Y^n) = W\}\,,
\end{equation}
and therefore, such codes are zero-error VLFT codes. 
\end{enumerate}

For both VLF and VLFT codes, we say that a code that satisfies  $\E{\tau} \leq \ell$ and $\Prob{W \neq \widehat W} \leq \epsilon$, when averaged over $U$, message and channel, is an $(\ell,  \epsilon)$ code for 
source/channel $\left( W, \{P_{Y_i | X^i Y^{i-1}}\}_{i = 1}^\infty \right)$. 

The random variable $U$ serves as the common randomness shared by both transmitter and receiver, which is used to
initialize the codebook. As a consequence of Caratheodory's theorem, the amount of this common randomness can always be
reduced to just a few bits: as shown in \cite[Theorem 19]{polyanskiy2011feedback}, if there exists an $(\ell,
\epsilon)$ code with $|\mathcal U| = \infty$, then there exists an  $(\ell,  \epsilon)$ code with $|\mathcal U| \leq 3$.
Allowing for common randomness does not affect the asymptotic expansions, but leads to more concise expressions for
our non-asymptotic achievability bounds. Furthermore, no common randomness is needed at all if the channel is symmetric \cite[Theorem 3]{polyanskiy2011feedback}.

First, we quote an achievability result~\cite[(107)-(118)]{polyanskiy2011feedback}. Let $P_Y$ be the capacity achieving output distribution of the DMC. Denote information density as usual:
	\begin{align}\label{eq:cx1}
		\imath_{X;Y}(a; b) &\eqdef \log {P_{Y|X}(b|a)\over P_Y(b)}.
	\end{align}	
	
\begin{thm}[\!\! \cite{polyanskiy2011feedback}]\label{th:X1} For every DMC with capacity $C$, any positive integer $M$ and probability of error
$\epsilon$ there exists an $(\ell, \epsilon)$ stop-feedback code for the message $W$ taking $M$ values\footnote{Although the result in \cite{polyanskiy2011feedback} is stated for average (over equiprobable messages) error probability, in fact, it applies to non-equiprobable W. Furthermore, if we do not insist on $|\mathcal{U}|\le 3$, Theorem  \ref{th:X1} also applies to
 maximal probability of error.} such that
	$$ C \ell \le \log M + {\log {1\over \epsilon}} + a_0\,, $$
	where 
\begin{equation}
 a_0 \triangleq \max_{ x,  y } 
 \imath_{X; Y}(x; y). \label{f2v:a0x}
\end{equation}
\end{thm}
We note that a similar result is also contained in many other works, starting from~\cite{burnashev1976data} and later
in~\cite{TT06,NJ13}.

Next, we tighten Theorem \ref{th:X1} in the case of non-equiprobable messages. 
\begin{thm}\label{th:corX1}
For every DMC with capacity $C$ and random variable $W$ there exists an $(\ell, \epsilon)$ stop-feedback code for $W$
with 
\begin{equation}\label{eq:corX1}
	C\ell \le H(W) + \log {1\over \epsilon} + a_0,
\end{equation}
where $a_0$ is defined in \eqref{f2v:a0x}.  
\end{thm}

\begin{proof}If $H(W)=\infty$ then there is nothing to prove, so we assume otherwise.
Denote the information in $m$
\begin{align}		
	\imath_{W}(m) &\eqdef \log {1\over P_W(m)}\,,
\end{align}	
and note that by the memorylessness assumption,
\begin{align}
	\imath_{X^n;Y^n}(a^n; b^n) &=\sum_{i=1}^n \imath_{X;Y}(a_i; b_i)\label{eq:cx2}.
\end{align}		 
In the achievability scheme of \cite[Theorem 3]{polyanskiy2011feedback}, at time $n$ the decoder observes $Y^n$, computes $M$ information densities $\imath_{X^n; Y^n}(\vect C^n_1; Y^n), \ldots, \imath_{X^n; Y^n}(\vect C^n_M; Y^n)$, where $\vect C^n_1, \ldots, \vect C^n_M$ are the codewords, and stops the first time one of the information densities exceeds a threshold.  However, instead of one common threshold, we assign lower thresholds to the more likely messages. 

\textit{Code construction:}
We define the common randomness (revealed to the encoder and decoder before the
transmission starts) to be a random variable $U$ as follows: 
\begin{equation}
P_U \triangleq  \underbrace{P_{X^\infty} \times \ldots  \times P_{X^\infty} }_{|\mathcal{W}| }, 
\end{equation}
where $X^\infty$ consists of i.i.d. copies drawn from (any) capacity-achieving input distribution.  
A realization of $U$ defines $|\matW|$ infinite dimensional vectors $\mathbf C_m^\infty \in \mathcal A^\infty$, 
$m \in \matW$. 
To transmit $m_0\in \matW$, the encoder passes the entries of the corresponding codeword $\vect C_{m_0}^\infty$ to the channel one by one, so that the first $n$ entries of the codeword, $\vect C_{m_0}^n$, are transmitted by time $n$:
\begin{align}
X^n = 
 \mathbf C_{m_0}^n. 
\end{align}

 At time $n$, the decoder computes the values
\begin{equation}
 I_{n}(m) =   \imath_{X^n; Y^n}(\mathbf C_{m}^n ; Y^n) - \imath_{W}(m ), 
\end{equation}
The decoder defines the stopping times:
\begin{equation}
 \tau_m = \inf\{n \geq 0 \colon I_{n}(m) \geq \gamma\},
\end{equation}
where $\gamma>0$ is an arbitrary constant.  The final decision $\widehat W$ is made by the decoder at the stopping time $\tau^\star$:
\begin{align}
 \tau^\star &= \min_{m \in \matW} \tau_m \label{eq:cx3}\\
 \widehat W &= g(Y^{\tau^\star}) = \argmin_{m \in \matW} \tau_m\,.
\end{align}
where the tie-breaking rule is immaterial. 

\textit{Analysis:} We claim that, averaged over $U$, we have:
	\begin{align}  \PP[W \neq \widehat W] &\le \exp(-\gamma) \label{eq:cx5}\\
	    C\, \EE[\tau^\star] &\le H(W) + \gamma + a_0\,.\label{eq:cx6}
\end{align}

Abbreviate the stopping time of the true message as
\begin{equation}
  \tau \eqdef \tau_W.
\end{equation}
The union bound results in, cf.~\cite[Theorem 3]{polyanskiy2011feedback}: 
\begin{align} \PP[W \neq \widehat W | W=m_0] & \le \sum_{m\in \matW\setminus\{m_0\}} \PP[\tau_m \le \tau |
W=m_0].
\label{eq:cx7}
\end{align}
Due to memorylessness of the channel, $\imath_{X^n; Y^n}(X^n; Y^n) - nC$ is a martingale with jump size upper-bounded by $a_0$ (defined in \eqref{f2v:a0x}). For each $j \in \mathbb N$, $\min\{\tau, j\}$ is a bounded stopping time, so
 by Doob's optional stopping theorem (e.g. \cite[Theorem 10.10]{williams1991probability}) we have
 \begin{align}
C\, \E{\min\{\tau, j\}}  - H(W) &= \E{I_{\min\{\tau, j\}}(W)} \leq \gamma + a_0 \label{f2v:a0}.
\end{align}
By the monotone convergence theorem, it follows from \eqref{f2v:a0} that
\begin{align}
C\, \E{\tau} - H(W) \leq \gamma + a_0 \label{f2v:Etauu}.
\end{align}
Therefore, the stopping time $\tau$ has bounded expectation, and since the martingale $\imath_{X^n; Y^n}(X^n; Y^n) - nC$ also has bounded increments, Doob's optional stopping theorem applies to conclude
\begin{equation}
\E{\imath_{X^\tau; Y^\tau}(X^\tau; Y^\tau) } = C\, \E{\tau}. 
\end{equation}

Next, we have, by a change of measure argument, for every $m\neq m_0$:
\begin{equation}
 \PP[\tau_m = n | W=m_0] \le \exp\left( -\imath_W(m) - \gamma\right) \PP[\tau_m = n | W=m]\,. \label{eq:cm}
\end{equation}
Consequently, using \eqref{eq:cm}, we have
\begin{align}
	\PP[\tau_m \le \tau | W=m_0] &\le \PP[\tau_m < \infty | W=m_0]\\
	&= \sum_{n = 0}^\infty \PP[\tau_m = n | W=m_0]\\
	&\le \exp\left(-\imath_W(m) - \gamma\right)\,,\label{eq:cx4}
\end{align}	
where we used that $\Prob{\tau < \infty}= \Prob{\tau_m <\infty | W=m} = 1$, which in turn follows from~\eqref{f2v:Etauu}.
Summing~\eqref{eq:cx4} over all $m\neq m_0$ and using~\eqref{eq:cx7} we get~\eqref{eq:cx5}. Note that the reasoning in \eqref{eq:cm}--\eqref{eq:cx4} generalizes that in~\cite[(111)-(118)]{polyanskiy2011feedback} to nonequiprobable messages.

The estimate of average length in~\eqref{eq:cx6} follows from $\tau^\star \le \tau $
and~\eqref{f2v:a0}. Finally, the desired bound~\eqref{eq:corX1} follows by taking $\gamma = \log {1\over \epsilon}$.
\end{proof}

\begin{remark}\label{rm:vlsep} A slightly less sharp bound could also be derived via a  
\textit{variable-length separated scheme}: compress $W$ losslessly into a variable-length string $\{0,1\}^*$ with
average length less than $H(W)$, cf.~\cite{wyner1972upper}, then send the length via $O(\log H(W))$ channel symbols incurring in a
very small probability of error and finally send the data bits.
\end{remark}

Next, we extend the zero-error bound in~\cite[Theorem 11]{polyanskiy2011feedback} to the case of
non-equiprobable messages:
\begin{thm} For every DMC with capacity $C$ there exists a constant $a_1$ such that for every discrete random variable $W$ there exists an $(\ell, 0)$
 VLFT code with
\begin{align}
  C \ell &\leq H(W) + a_1. \label{eq:X3}
\end{align}
\label{th:X3}
\end{thm}
\apxonly{We can also prove: For every $Q_W$ we have
$$ C\ell \le \EE{\log {1\over Q_W(W)}} + a_1 $$
This could be useful to fast-track variable-length proofs: Just choose $Q_W(m) = c \exp\left(-\lambda \cdot \ell(m)\right)$
}
\begin{proof}
Without loss of generality, we assume that $H(W)<\infty$ and $W$ takes values in positive integers.
The codebook is a countable collection of infinite strings $\mathbf C_m^\infty$, $m = 1, 2, \ldots$. Given the codebook and
the realization of $W=m_0$, the encoder sends  in the first $n$ channel uses
\begin{equation}
X^n = \mathbf C_{m_0}^n.
\end{equation}
The decoder outputs $m$ error-free at the stopping time $\tau^\star$ given by:
\begin{equation}
 \tau^\star = \inf_n \left\{ I_n(m) > \max_{j \neq m} I_n(j) \right\} \label{eq:taustar},
\end{equation}
where
\begin{equation}
 I_{n}(m) = \imath_{X^n; Y^n}(\mathbf C_m^n; Y^n) - \imath_{W}(m).
\end{equation}
According to \eqref{eq:taustar}, if the true message is $m$ the transmission stops at the first instant $n$ when $I_n(m)$ exceeds all $I_n(j)$, $j \neq m$. Note that $\tau$ depends on the transmitted message $m$, as permitted by the  paradigm of VLFT codes.

{\it Analysis:}
The probability that the time to transmit message $m$ exceeds $n$ is 
\begin{align}
\Prob{\tau^\star > n | W = m} = \Prob{ \bigcup_{j \neq m} \left\{ I_n(j) > I_n(m) \right\}| W = m} \label{eq:tauexcess}.
\end{align}
Applying the random coding argument, we now assume that the codebook strings $\mathbf C_1^\infty, \mathbf C_2^\infty,
\ldots$ are drawn i.i.d. from $P_{X^\infty} = P_{\mathsf X} \times P_{\mathsf X} \times \ldots$, where $P_{\mathsf X}$ is the capacity-achieving channel input distribution.
Denoting by $\bar X^n$ an independent copy of $X^n$ and taking the expectation of the right side of \eqref{eq:tauexcess} with respect to the codebook, we have \eqref{f2v0:-A10}--\eqref{f2v0:-A1} at the bottom of the page, where we used the union bound and $\min\{1, a\} = \exp \left(- \left| \log \frac 1 a \right|^+\right)$.

\begin{table*}[!b]
\newcounter{mytempeqncnt}
\normalsize
\setcounter{mytempeqncnt}{\value{equation}}
\setcounter{equation}{\value{mytempeqncnt}}
\vspace*{4pt}
\hrulefill
 \begin{align}
 \Prob{\bigcup_{j \neq m} \left\{ \frac{ P_{W}(j) P_{Y|X}(Y|X_j)}{ P_{W}(m) P_{Y|X}(Y|X_m)} \geq 1\right\} | W = m } 
 \leq&~ 
 \mathbb E \Bigg[ \min\Bigg\{ 1, 
 \sum_{j = 1}^\infty \Prob{ \frac{ P_{W}(j) P_{Y^n|X^n}(Y|\bar X^n)}{ P_{W}(m) P_{Y^n|X^n}(Y^n|X^n)} \geq 1 \mid X^n, Y^n}\Bigg\}
 \Bigg] \label{f2v0:-A10}\\
 \leq&~ \E{ \min\left\{ 1, \sum_{j = 1}^{\infty} \frac{P_{W}(j)}{P_{W}(m )} \frac{\E{ P_{Y^n| X^n}(Y^n |\bar X^n )|Y^n }}{ P_{Y^n |X^n }(Y^n |X^n )}\right\} } \label{jscc:-Ac}\\
= &~ \E{\exp(- | \imath_{X^n; Y^n}(X^n; Y^n) - \imath_{W}(m)|^+)} \label{f2v0:-A1}, 
 \end{align}
\end{table*}

Applying \eqref{f2v0:-A1} we get 
 \begin{align}
\E{\tau^\star} &= \sum_{n = 0}^\infty \Prob{\tau^\star > n}\\
	      &\le  \sum_{n = 0}^\infty \E{\exp(- | \imath_{X^n; Y^n}(X^n; Y^n) - \imath_{W}(W)|^+)} \label{f2v0:Ag}.
\end{align}
Finally,~\eqref{f2v0:Ag} implies~\eqref{eq:X3} by applying the
result~\cite[(162)-(179)]{polyanskiy2011feedback}:
\begin{equation}
  \sum_{n=0}^\infty \E{\exp(- | \imath_{X^n; Y^n}(X^n; Y^n) - \gamma|^+)} \le {\gamma\over C} + a_1\,,
\end{equation}
where $a_1$ is a constant determined by the distribution of $\imath_{X_1; Y_1}(X_1; Y_1)$. 
\end{proof}

\section{Asymptotic expansions of the rate-distortion tradeoff}\label{sec:main}

\subsection{Definitions}\label{sec:def}
We proceed from the setup of Section \ref{sec:newcodes} where a discrete message is transmitted over the channel with feedback to a more general scenario, in which a, possibly analog, signal is transmitted over a channel with feedback, under a fidelity constraint. We will consider the following scenarios:
\begin{enumerate}
\item \textit{excess distortion probability:} A VLF code transmitting memoryless 
source $S^k \in \mats^k$ with reproduction alphabet $\hat \mats^k$ and separable distortion measure $\mathsf d \colon \mathcal S^k \times \hat {\mathcal S}^k \mapsto [0, + \infty]$
is called a $(k, \ell, d, \epsilon)$ excess-distortion code if the decoding time and the distortion satisfy
	\begin{align} \EE[\tau] &\le \ell\label{eq:def0},\\
	   \PP[\sfd(S^k, \hat S^k) > d] &\le \epsilon \label{eq:def1}.
	\end{align}
The corresponding fundamental limit is 
\begin{equation}
 \ell^\star(k, d, \epsilon) \eqdef \inf \{ \ell \colon \exists \text{ an $(k, \ell, d, \epsilon)$ VLF code}\}\,.
\end{equation}

\item \textit{average distortion:} A VLF code satisfying, instead of~\eqref{eq:def1}, an average constraint
	\begin{equation}\label{eq:def2}
		\EE[\sfd(S^k, \hat S^k)] \le d. 
\end{equation}	
	is called a $(k, \ell, d)$ average-distortion code. The corresponding fundamental limit is 
\begin{equation}
 \ell^\star(k, d) \eqdef \inf \{ \ell \colon \exists \text{ an $(k, \ell, d)$ VLF code}\}\,.
\end{equation}

\item \textit{guaranteed distortion}: A VLFT code transmitting memoryless 
source $S^k \in \mats^k$ with reproduction alphabet $\hat \mats^k$ and separable distortion metric $\mathsf d$
is called a $(k, \ell, d, 0)$ guaranteed-distortion code if it achieves $\epsilon=0$ in~\eqref{eq:def1}.
The corresponding fundamental limit is
\begin{equation}
 \ell^\star_t(k, d, 0) \eqdef \inf \{ \ell \colon \exists \text{ an $(k, \ell, d, 0)$ VLFT code}\}\,.
\end{equation}
\end{enumerate}

We will use the following notation for the various minimizations of information measures:
 \begin{align}
 {\mathbb R}_{S}(d) &\triangleq  \min_{ \substack{ P_{Z | S } \colon \mathcal S \mapsto \hat{\mathcal S}\colon  \\ \E{ \mathsf d( S, Z)} \leq d}} I(S; Z) \label{RR(d)},\\
  {\mathbb R}_{S}(d, \epsilon) &\triangleq  \min_{ \substack{ P_{Z | S } \colon \mathcal S \mapsto \hat{\mathcal S}\colon \\ \Prob{ \mathsf d( S, Z) > d} \leq \epsilon}} I(S; Z) \label{RR(eps)},\\
  H_{d, \epsilon}(S) &\triangleq \min_{\substack{ \mathsf c \colon \mathcal S \mapsto \hat {\mathcal S} \colon \\ \Prob{\mathsf d(S, \mathsf c (S) ) > d} \leq \epsilon}} H(\mathsf c (S)).  \label{eq:Hepsdelta}
\end{align}

The quantity in \eqref{eq:Hepsdelta} is referred to as the $\left(d, \epsilon \right)$-entropy of the source $S$ \cite{posner1967epsilonentropy}. 
%
%
The $(d, 0)$-entropy is also known as epsilon-entropy~\cite{posner1967epsilonentropy}:\footnote{N.B. in that terminology ``epsilon" corresponds to $d$, not $\epsilon$.}
\begin{equation}
 H_{d,0}(S) \triangleq \min_{\substack{ \mathsf c \colon \mathcal S \mapsto \hat {\mathcal S} \colon \\ \mathsf d(S,
\mathsf c (S) ) \leq d \text{ a.s.}}} H(\mathsf c (S)).  \label{eq:Heps}
\end{equation}

Provided that the infimum in \eqref{RR(d)} is achieved by some transition probability kernel $P_{Z^\star | S}$, the $\mathsf d$-tilted information in $s \in \mathcal S$ is defined as  \cite{kostina2011fixed}
\begin{align}
\jmath_S(s, d) &\triangleq - \log \E{\exp \left( - \lambda^\star \mathsf d(s, Z^\star) + \lambda^\star d \right) } \label{jE},
\end{align}
where\footnote{Note that the existence of $P_{Z^\star | S}$ guarantees the differentiability of $\mathbb R_S(d)$. } 
\begin{equation}
\lambda^\star = - \mathbb R^\prime_S (d).
\end{equation}
Note that in almost-lossless compression, $Z^\star = S$ and 
\begin{align}
 \jmath_S(s, d) &\triangleq \imath_S(s). 
\end{align}

\subsection{Regularity assumptions on the source}

We assume that the source, together with its distortion measure, satisfies the following assumptions: 
\begin{enumerate}[{A}1]
\item The source $\{S_i\}$ is stationary and memoryless,  $P_{S^k}  = P_{\mathsf S} \times \ldots \times P_{\mathsf S}$. \label{item:first}
\item The distortion measure is separable, $\mathsf d(s^k, z^k) = \frac 1 k \sum_{i = 1}^k \mathsf d(s_i, z_i)$. \label{item:separable}
\item The distortion level satisfies $d_{\min} < d < d_{\max}$, where $d_{\min}$ is the infimum of values at which 
 the minimal mutual information quantity ${\mathbb R}_S(d)$ 
is finite, and $d_{\max} =\inf_{\mathsf z \in \widehat {\mathcal S} } \E{\mathsf d(\mathsf S, \mathsf z)}$, where the expectation is with respect to the unconditional distribution of $\mathsf S$. 
 \label{item:dminmax}
 \item The rate-distortion function is achieved by some $P_{\mathsf Z^\star | \mathsf S}$: $\mathbb R_{\mathsf S}(d) = I(\mathsf S; \mathsf Z^\star)$. 
 \item $
 \E{{\mathsf d}^{12}(\mathsf S, \mathsf Z^\star)} < \infty
$  
where the expectation is with respect to $P_{\mathsf S} \times P_{\mathsf Z^\star}$. \label{item:last}

\end{enumerate}

The rate-dispersion function of the source satisfying assumptions A\ref{item:first}--A\ref{item:last} is given by \cite{kostina2011fixed}
\begin{equation}
\mathcal V(d) = \Var{\jmath_{\mathsf S}(\mathsf S, d)}.  
\end{equation}
 
We showed in \cite{kostina2015varrate} that under assumptions A\ref{item:first}--A\ref{item:last} for all $0 \leq \epsilon \leq 1$
 \begin{equation}
 \left.\begin{aligned}
        & {\mathbb R}_{S^k}(d, \epsilon)\\
         & H_{d, \epsilon}(S^k)
       \end{aligned}
 \right\}
 = (1 - \epsilon) k R(d) -   \sqrt{\frac{k \mathcal V(d)}{2 \pi} } \mathrm{e}^{- \frac { (\Qinv{\epsilon})^2} 2 }  +\bigo{\log k}  \label{minI2orderlossy} .
\end{equation}

\subsection{Average distortion}

\begin{thm}\label{th:A} Under assumptions A\ref{item:first}--A\ref{item:last} we have
\begin{equation}
 C\ell^*(k, d) = k R(d) + O(\log k)\,. \label{eq:A}
\end{equation}
\end{thm}
\begin{proof}

\textit{Achievability:}  fix $1 < p \leq 12$, $\epsilon > 0$, source codebook  $c_1, \ldots, c_M$ and consider a separated scheme $$S - \mathsf f(S) - X - Y - \widehat W - \mathsf c(\widehat W)$$ such that $\Prob{\mathsf f(S) \neq \widehat W} \leq \epsilon$ and 
\begin{align}
\mathsf f(s^k) &= \arg \min_{m \in \left\{1, \ldots, M\right\} } \mathsf d(s^k, c_m), \quad s^k \in \mathcal S^k\\
\mathsf c(m) &= c_m , \quad m \in \left\{1, \ldots, M\right\}
\end{align}
The average distortion is bounded by
\begin{align}
&~\E{\mathsf d \left(S^k,  c_{\widehat W } \right) } \notag\\
\leq&~ \E{\mathsf d(S^k, c_{W})  }  + \E{\mathsf d(S^k, c_{\widehat W}) \1{\widehat W \neq W}}\\
\leq&~ \E{ \min_{m \in \left\{1, \ldots, M\right\} } \mathsf d(S^k,  c_m)   }+ \left( \E{\mathsf d^{p}(S^k, c_{\widehat W}) }\right)^{\frac 1 p} \epsilon^{1 - \frac 1 p}  \label{holder}
\end{align} 
where \eqref{holder} holds by H\"{o}lder's inequality. Taking the expectation of both sides over $c_1, \ldots, c_M$ drawn i.i.d. from $P_{\mathsf Z^\star}^k$, we conclude, via a random coding argument, that there exists a separate source/channel code with average distortion bounded by
\begin{equation}
d \leq \bar d  + \mathbb E^{\frac 1 p} \left[  \mathsf d^p (S^k,  Z_1)  \right]  \epsilon^{1 - \frac 1 p} \label{h2},
\end{equation}
where we denoted
\begin{equation}
\bar d \triangleq \E{\min_{m \in \left\{1, \ldots, M\right\} } \mathsf d(S^k,  Z_m )  }, 
\end{equation}
and $P_{S^k  Z_1\ldots Z_M } = P_{S^k} P_{\mathsf Z^\star}^k \ldots P_{\mathsf Z^\star}^k$. The work of  Pilc \cite{pilc1967transmission} (finite alphabet) and Yang and Zhang
\cite{yang1999redundancy} (abstract alphabet) shows that under assumptions A\ref{item:first}--A\ref{item:last},\footnote{In particular, the bounded twelfth moment in A\ref{item:last} is required for the applicability of the result of Yang and Zhang \cite{yang1999redundancy}. } 
\begin{equation}
 \log M = kR\left(\bar d\,\right) + O(\log
k).
\end{equation}
 Letting $\epsilon = k^{ \frac {p} {p -1}}$, 
we conclude by assumption A\ref{item:last} that the second term in \eqref{h2} is bounded by $\bigo{\frac 1 k}$, i.e. $d \leq \bar d + \bigo{\frac 1 k}$. Finally, by Theorem~\ref{th:X1}$, C \ell \leq \log M + \bigo{\log k} $, and the `$\leq$' direction in \eqref{eq:A} follows by the differentiability of $R(d)$ in the region of assumption A\ref{item:dminmax}. 

\textit{Converse:}  By the data-processing inequality and~\cite[Lemma 1-2]{burnashev1976data} we have
\begin{equation}
 kR(d) \le \ell C
\end{equation}
	for any $(k, \ell, d)$ VLF code.
\end{proof}

\subsection{Excess distortion probability}


\begin{thm}\label{th:B} Under assumptions A\ref{item:first}--A\ref{item:last} and any $\epsilon>0$ we have
 \begin{equation}
\ell^\star(k, d, \epsilon)\, C  =  (1 - \epsilon) k R(d)  -  \sqrt{\frac{k \mathcal V(d)}{2 \pi} } \mathrm{e}^{- \frac { (\Qinv{\epsilon})^2} 2 }  + \bigo{\log k} \label{f2v:2order}
\end{equation}
\label{f2v:thm:2order}
\end{thm}
\begin{proof}
\textit{Achievability:} Pair a lossy compressor $S^k \to W$ with excess-distortion probability 
$\epsilon' = \epsilon - {1\over \sqrt{k}}$ and $H(W)=H_{d,\epsilon'}(S^k)$\footnote{Although the optimal mapping $\mathsf c^\star$ that achieves $\left(d, \epsilon\right)$-entropy is not known in general,
the existence of good approximations satisfying the constraint in \eqref{eq:Hepsdelta} that approach $H_{d,
\epsilon}(S^k)$ to within $\log_2 H_{d, \epsilon}(S^k)$ bits is guaranteed by a random coding argument, see \cite{kostina2015varrate}.} with a VLF code from Theorem~\ref{th:corX1}
transmitting $W$ with probability of error ${1\over \sqrt{k}}$.
Apply \eqref{minI2orderlossy} to \eqref{eq:corX1}\footnote{Note that \eqref{minI2orderlossy} also holds if $\epsilon$ in the left side is replaced by $\epsilon + \bigo{\frac 1 {\sqrt k}}$.}. 

\textit{Converse:} Apply the data-processing inequality and~\cite[Lemma 1-2]{burnashev1976data} to get:
\begin{equation}
  \ell C \ge {\mathbb R}_{S^k}(d, \epsilon)
\end{equation}
for every $(k, \ell, d, \epsilon)$ VLF code.
\end{proof}

\subsection{Guaranteed distortion}

\begin{thm}\label{th:C} Under assumptions A\ref{item:first}--A\ref{item:last}, we have
\begin{equation}
\ell^\star_t(k, d, 0)\, C = k R(d) + \bigo{\log k} \label{f2v0:2order0}
\end{equation}
\label{f2v0:thm:2order}
\end{thm}
\begin{proof}For the achievability we note that the 
estimate of the $H_{d,\epsilon}(S^k)$ in~\eqref{minI2orderlossy} applies with $\epsilon=0$ and thus
\begin{equation}
 H_{d,0}(S^k) = k R(d) + O(\log k)\,.
\end{equation}
Then, we can pair the mapping achieving $H_{d,0}(S^k)$ with the zero-error VLFT code from Theorem~\ref{th:X3}. 

Conversely, repeating the argument of \cite[Theorem 4]{polyanskiy2011feedback}, with the replacement of the right side
of \cite[(67)]{polyanskiy2011feedback} by $\mathbb R_S(d, \epsilon)$ we conclude that any $(\ell, d, \epsilon)$ VLFT code must satisfy
\begin{equation}
{\mathbb R}_{S}(d, \epsilon)\leq C \ell + \log (\ell + 1) + \log e\,. \label{f2v:C0}
\end{equation}
\end{proof}

\subsection{Discussion}

We make several remarks regarding the rate-distortion tradeoff in all three settings considered
above:
\begin{enumerate}
\item The case $d = d_{\min}$ is excluded by the assumptions of Theorems \ref{th:A}--\ref{th:C}. However, in the important special case of a distortion measure that satisfies 
\begin{equation}
 \mathsf d(a, b) = 
\begin{cases}
   d_{\min}, & a = b\\
 > d_{\min}, & a \neq b
\end{cases}\, ,
\end{equation}
$d = d_{\min}$ corresponds to almost-lossless transmission, and both Theorems \ref{th:A} and \ref{th:B} apply with $R(d)$ and $\mathcal V(d)$ equal to the entropy and the varentropy of the source, respectively, as long as the source is stationary and memoryless and the third moment of $\imath_{\mathsf S}(\mathsf S)$ is finite. 

\item
For almost-lossless transmission of finite alphabet sources, the asymptotic expansion \eqref{f2v:2order} can be achieved by
reliably (i.e. with probability of error $\sim \frac 1 k$) transmitting the type of the source outcome first followed by an index that describes
the source outcome within its type class, unless the type of the source outcome is one of the least likely types with total mass $\epsilon$, in which case nothing is transmitted. 

\item Even if the channel is not symmetric, the asymptotic expansions in Theorems \ref{th:A}--\ref{th:C} can be achieved without common randomness $U$, by using constant composition channel codebooks.  For instance, consider the scheme in Theorem~\ref{th:corX1} with $P_{X^\infty}$ drawn from the distribution equiprobable over the capacity-achieving type. Since $\E{\tau | X^\infty} = \E{\tau}$ a.s., almost every such codebook attains the bound in \eqref{f2v:a0} up to logarithmic terms, resulting in a deterministic construction attaining \eqref{f2v:2order}. 

\item Stop-feedback codes are remarkably powerful at finite blocklength; indeed, up to the terms of order $\bigo{ \log k}$, they attain the fundamental limits in the settings of Theorems \ref{th:A}, \ref{f2v:thm:2order} and \ref{th:C}. As the converse parts of Theorems \ref{th:A}, \ref{f2v:thm:2order} and \ref{th:C} demonstrate, relying on feedback more heavily can only bring in an improvement of order at most $\bigo{\log k}$.

\item 
Note that \eqref{f2v:2order} is achieved by a stop-feedback code. We can further show that even without any feedback one
can still achieve the optimal first-order performance 
\begin{equation}
  \ell C \le (1-\epsilon) k R(d) + O(\sqrt{k \log k}), \label{f2v:avnof}
\end{equation}
provided variable-length channel coding is allowed. Indeed, one can first use the variable-length excess-distortion compressor
from~\cite{kostina2015varrate} on $S^k$ to get a binary string of average length $(1-\epsilon) k R(d) + O(\sqrt{k})$,
see~\eqref{intro:s}. Then, truncating the length at $k^2$ and transmitting $2\log k$ data bits with reliability ${1\over
k^2}$, we can reliably inform the encoder about the total number of data bits $b$ to be sent next. We may then use a
capacity-achieving code of length ${b\over C} + O(\sqrt{b \log b})$ to send the data bits with reliability $1\over k$ \cite{polyanskiy2010channel}.

\item The naive separation scheme, i.e. a fixed-length source code followed by a channel code achieves at most:
	\begin{equation}\label{eq:dc1}
		\ell C \ge (1-\epsilon) kR(d) + a  \sqrt{k \log k},\, a>0 .
\end{equation}	
Indeed, according to Theorem \ref{th:X1}, the number of messages $M$ that can be transmitted via a
VLF code with error probability $\eta$ satisfies
\begin{equation}
 \log M \geq \frac{\ell C}{1 - \eta} + \bigo{\log \ell} \label{f2v:2orderchannel}.
\end{equation}
On the other hand, the number of codewords of a source code with probability of exceeding distortion $d$ no greater than $\zeta$ satisfies \cite{kostina2011fixed} 
\begin{equation}
 \log M \leq k R(d) + \sqrt {k \mathcal V(d)} \Qinv{\zeta} + \bigo{1} \label{f2v:2ordersource}.
\end{equation}
Optimizing over $\eta + \zeta \le \epsilon$ yields~\eqref{eq:dc1}. \apxonly{
Optimizing $\zeta = \frac 1 {\sqrt k}$ and using $\Qinv{\zeta} \sim \sqrt{2 \log_e \frac 1 \zeta}$ as $\zeta \to 0$.} 

\item The semi-joint separated schemes that attain~\eqref{f2v:2order}   contain a vital ingredient missing from naive separated schemes: namely, the channel code employs unequal error protection. Consequently, the more likely source codewords are decoded with higher reliability, resulting in massive improvement at finite blocklength evidenced by ~\eqref{f2v:2order}.  Unequal error protection can be achieved ether via a maximum-a-posteriori-like decoder of Theorem \ref{th:corX1} or the variable-length separation interface of Remark~\ref{rm:vlsep}.

\item The Schalkwijk-Bluestein  \cite{schalkwijk1967transmission} (see also \cite{cruise1967achievement})
elegant linear feedback scheme for the transmission of a single Gaussian sample $ S \sim \mathcal N(0, \sigma^2)$ over the AWGN channel achieves the mean-square error $\frac {\sigma^2}{\left( 1 + P\right)^n}$, after $n$ channel uses, where $P$ is the average transmit SNR.  In other words, the minimum delay in transmitting a Gaussian sample over a Gaussian channel with feedback is given by
\begin{equation}
  \ell^\star(1, d) = \frac {R(d)}{C}, \label{elias}
\end{equation}
as long as $\frac {R(d)}{C}$ is integer.\footnote{If $R(d) = C$, no feedback is needed. } Note that \eqref{elias} is achieved with fixed, not variable length, and average, not maximal, power constraint. If there are $k$ Gaussian samples to transmit, repeating the scheme for each of the samples achieves
\begin{equation}
  \ell^\star(k, d) = k \frac {R(d)}{C},  \label{elias2}
\end{equation}
which implies, in particular, that in general our estimate of $\bigo{\log k}$ in \eqref{eq:A} is too conservative. Beyond Gaussian sources and channels, a sufficient condition for a fixed-length JSCC feedback scheme to achieve  \eqref{elias2} is provided in \cite{gastpar2003source}. 

\item  The Schalkwijk-Bluestein scheme uses instantaneous feedback and has notoriously resisted generalization beyond Gaussian channels, which limits the applicability of the scheme. In contrast, the simple separated scheme in Theorem \ref{th:A} uses only stop-feedback and applies to arbitrary sources and channels. 
\end{enumerate}

\apxonly{
Alternative achievability proofs:
\begin{proof}[Achievability proof A: type splitting]
First, we make the following observations. 
\begin{enumerate}
 \item By \cite[(101)]{polyanskiy2011feedback} and the observation that although its proof is focused on average error probability, it in fact applies to the worst case error probability, there exists a VLF code capable of distinguishing all source types with average delay $\ell^\prime$ and probability of error $\frac 1 {\ell^\prime}$ such that 
\begin{equation}
|\mathcal S|\log\left( k + 1 \right) \geq C \ell^\prime - \log \ell^\prime - a_0 \label{eq:Atype}
\end{equation}
where $a_0 > 0$ is a constant. 
 \item By \cite[Theorem 2]{polyanskiy2011feedback}, there exists a VLF code capable of distinguishing all source outcomes of type $p$ with probability of error $\epsilon_p$ whose average delay $\ell_p$ satisfies 
\begin{equation}
 k H(p)(1 - \epsilon_p) \geq  C \ell_p - \bigo{\log \ell_p} \label{eq:Aseq}
\end{equation}
\end{enumerate}

Let
\begin{equation}
\epsilon_P = \epsilon_P^\star - \frac{1}{\ell^\prime}
\end{equation}
where $\epsilon_P^\star$ is the error profile that achieves the minimum in \eqref{eq:minIpe}. Having observed the source outcome $s^k$ of type $p$, the encoder concatenates the codes above to send the type of the sequence followed by the index of $s^k$ in the list of all sequences of type $p$. The average error probability and average delay of this code are given by
\begin{align}
\epsilon &= \E{\epsilon_P}\\
\ell &= \E{\ell_P} + \ell^\prime 
\end{align}
and \eqref{f2v:2order} follows using \eqref{eq:Atype} and averaging \eqref{eq:Aseq} as in the proof B of Theorem \ref{thm:minIpe}.
\end{proof}
\begin{proof}[Achievability proof B: length coding] 
The advantage of this method is its applicability to countably infinite source alphabets. 
Abbreviate
\begin{equation}
1\left\{ x \leq_\alpha \eta \right\} = 1\left\{ x <_\alpha \eta \right\} + \alpha \left\{ x = \eta \right\}
\end{equation}
Let $\eta > 0$ and $0 \leq \alpha < 1$ be the unique solution to 
\begin{equation}
\epsilon = \Prob{P_{S}(S) \leq_\alpha \eta}  \label{eq:lemmaRda}
\end{equation}

First, the source outcome is encoded using the optimal variable-length code that achieves probability of error $\epsilon$. Then, the output of the variable-length encoder is transmitted through the feedback channel using the following observations.  
\begin{enumerate}
 \item By \cite[(101)]{polyanskiy2011feedback} and the observation that although its proof is focused on average error probability, it in fact applies to the worst case error probability, there exists a VLF code capable of distinguishing all encoded lengths with average delay $\ell^\prime$ and probability of error $\frac 1 {\ell^\prime}$ such that

\begin{equation}
\log \lfloor \log M_S^{+} (\eta) \rfloor \geq C \ell^\prime - \log \ell^\prime - a_0 \label{eq:Alen}
\end{equation}
where $a_0 > 0$ is a constant, and $M_S^+(\eta)$ is the number of masses whose probability is greater than or equal to $\eta$.
\item Since the number of sequences that has the same encoded length as $S^k$ is $\exp\left( \lfloor \log S^k \rfloor \right)$, by \cite[(101)]{polyanskiy2011feedback} there exists an $(\ell_{S^k}, \exp\left( \lfloor \log S^k \rfloor \right), \frac 1 {\ell_{S^k}})$ VLF code capable of distinguishing all source outcomes of   length equal to $\left \lfloor \log S^k \right \rfloor$
\begin{equation}
\left \lfloor \log S^k \right \rfloor \geq  C \ell_{S^k} - \log \ell_{S^k} - a_0  \label{eq:Aseqsamelen}
\end{equation}
\end{enumerate}
Taking the expectation of \eqref{eq:Aseqsamelen} over $S^k$, we obtain
\begin{equation}
\E{\left \lfloor \log S^k \right \rfloor 1 \left\{P_{S^k}(S^k) \geq_{1 - \alpha} \eta \right\} }  \geq  C \ell - \log \ell
\end{equation}
Since the left side is equal to 
\begin{equation}
(1 - \epsilon) k H(\mathsf S) -   \sqrt{\frac{k V(\mathsf S)}{2 \pi} } \mathrm{e}^{- \frac { (\Qinv{\epsilon})^2} 2 } + \bigo{1} 
\end{equation}
 \eqref{f2v:2order} follows.
 \end{proof}
}



\section{Energy-limited feedback codes for non-equiprobable messages}\label{sec:Ecodes}
In this section, we study the transmission of a message over an AWGN channel under an energy constraint. We would like to know how much information can be pushed through the channel, if a total of $E$ units of energy is available to accomplish the task. Formally, the codes studied in this section are  defined as follows. 
\begin{defn}
An energy-limited code without feedback for the transmission of a random variable $W$ taking values in $\matW$ over an AWGN channel is defined by:

\begin{enumerate}

\item A sequence of encoders
$\mathsf f_n \colon \matW \mapsto \mathcal A$, specifying the channel inputs
\begin{equation}
X_n = \mathsf f_n \left( W \right) 
\end{equation}
satisfying
\begin{equation}
\Prob{ \sum_{j=1}^\infty X_j^2 \le E} = 1. \label{E:Emax}
\end{equation}

\item A decoder $\mathsf g \colon \mathcal B^\infty \mapsto  {\mathcal W}$ producing $\hat W = g(Y^\infty)$, where
$\{Y_j\}$ is the output of a memoryless AWGN channel:
	\begin{equation}\label{eq:awgn_channel}
		Y_j = X_j + Z_j, \qquad Z_j \sim \matn\left(0, {N_0\over 2}\right)\,.
\end{equation}	
\end{enumerate}

\label{E:def}
 \end{defn}

\begin{defn}
An energy-limited feedback code for the transmission of a random variable $W$ taking values in $\matW$ over an AWGN channel is defined by:
\begin{enumerate}

\item A sequence of encoders
$\mathsf f_n \colon \matW \times \mathcal B^{n - 1} \mapsto \mathcal A$, defining the channel inputs
\begin{equation}
X_n = \mathsf f_n \left( W, Y^{n-1}\right) 
\end{equation}
satisfying
\begin{equation}
 \sum_{j=1}^\infty \E{X_j^2} \leq E. \label{E:Eav}
\end{equation}

\item A decoder $\mathsf g \colon \mathcal B^\infty \mapsto  {\mathcal W}$ producing $\hat W = g(Y^\infty)$, where
$\{Y_j\}$ is the output of the memoryless AWGN channel~\eqref{eq:awgn_channel}.
\end{enumerate}
\label{E:deff}
 \end{defn}

An $(E,  \epsilon)$ code for the transmission of random variable $W$ over the Gaussian channel  is a code with energy bounded by $E$ and $\Prob{W \neq \widehat W} \leq \epsilon$. 

Definitions \ref{E:def}--\ref{E:deff} do not impose any restrictions on the number of degrees of freedom $n$, restricting instead the total available energy.  
The problem of transmitting a
message with minimum energy was posed by  Shannon \cite{shannon1949communication}, who showed that $E_b$, the minimum energy per information bit compatible with vanishing block error probability converges to $N_0 \log_e 2$ as the number of information bits goes to infinity, where $\frac {N_0}{2}$ is the noise power per degree of freedom.  
Recently, Polyanskiy et al. \cite[Theorem 7]{polyanskiy2011minimumenergy} showed a dynamic programming algorithm for the error-free transmission of a single bit over an AWGN channel with feedback that attains {\it exactly} Shannon's optimal energy per bit tradeoff 
\begin{equation}
   E_b = N_0 \log_e 2. \label{E:pol0}
\end{equation}
The next non-asymptotic achievability result leverages that algorithm to transmit error-free a binary representation of a random variable over the AWGN channel by means of a variable-length separate compression/transmission scheme.
\begin{thm}
There exists a zero-error feedback code for the transmission of  a random variable $W$ over the AWGN channel with energy 
\begin{equation}
\frac{E}{N_0} \log e< H(W) + 1  \label{E:Af}.
\end{equation} 
Conversely, any $(E,0)$-feedback code must satisfy
\begin{equation}
 H(W) \leq \frac{E}{N_0} \log e. \label{E:Cf}
\end{equation}

\label{thm:ECfeedback}
\end{thm}
\begin{proof}
The encoder converts the source into a variable-length string using the Huffman code, so that the codebook is prefix-free and the expectation of the encoded length $\ell(W)$ is bounded as
\begin{equation}
\E{ \ell(W) }  < H(W) + 1\label{E:Af1}\,.
\end{equation}
Next, each bit (out of $\ell(W)$) is transmitted at the optimal energy per bit tradeoff $N_0 \log_e 2$ using the
zero-error feedback scheme in \cite[Theorem 7]{polyanskiy2011minimumenergy}. Transmissions corresponding to different bits are
interleaved diagonally (see Fig.~\ref{fig:diag}): the first bit is transmitted in time slots $1, 2, 4, 7, 11,
\ldots$, the second one in $3, 5, 8, 12, \ldots$, and so on. The channel encoder is silent at those indices allocated to source bits $\ell (W) +1, \ell(W) +2, \ldots$ For example, if the codeword has length 2 nothing is transmitted in time slots
$6, 9, 13, \ldots$. 
The receiver decodes the first transmitted bit focusing on the time slots $1, 2, 4, 7, 11, \ldots$ It proceeds
successively with the second bit, etc., 
until it forms a codeword of the Huffman code, at which point it halts.
Thus, it does not need to examine the outputs of the time slots corresponding to information bits that were not transmitted, and in which the encoder was silent.

\begin{figure}
\def\hhh{\hskip 5pt}
\centering
\begin{tabular}{r|ccccc}
Bit number & \multicolumn{5}{c}{Sequence of time slots}\\
\hline
1 & \hhh1\hhh & \hhh2\hhh & \hhh4\hhh  & \hhh7\hhh & $\cdots$ \\
2 & \hhh3\hhh & \hhh5\hhh & \hhh8\hhh & $\cdots$ \\
3 & \hhh6\hhh & \hhh9\hhh & $\cdots$ \\
$\vdots$ & \\
$\ell(W)$ & 
\end{tabular}
\caption{Illustration of the diagonal numbering of channel uses in Theorem~\ref{thm:ECfeedback}.}\label{fig:diag}
\end{figure}

Since the scheme spends $N_0 \log_e 2$ energy per bit, the total energy to transmit the codeword representing $W$ is
\begin{equation}
 \ell(W) N_0 \log_e 2 \label{E:Af2}.
\end{equation}
Taking the expectation of \eqref{E:Af2} over $W$ and applying \eqref{E:Af1}, \eqref{E:Af} follows. 

 In the converse direction, due to the zero-error requirement and data processing, $H(W) = I(W; \mathsf g(Y^\infty)) \le
 I(W; Y^\infty)$. Let $\Phi = \prod_{k=1}^\infty \matn\left(0, {N_0\over 2}\right)$ be a product measure on
 $\mreals^\infty$. It was shown in the proof of~\cite[Theorem 4]{polyanskiy2011minimumenergy} that the joint
 distribution of 
 	$$ \left( \log_e {dP_{Y^\infty|W}(Y^\infty|W)\over d\Phi(Y^\infty)},~  \sum_{k=1}^\infty X_k^2\right)\,
 	$$ 
coincides with the joint distribution of 
	$$ \left( {2\over N_0} \cdot B_\tau,~ \tau\right)\,,$$
where $\tau$ is the stopping time of the Brownian motion $B_t = {t\over 2} + \sqrt{N_0\over 2} W_t$ defined in the proof of~\cite[Theorem 4]{polyanskiy2011minimumenergy}, and $W_t$ is the
standard Wiener process. From Doob's optional stopping theorem we obtain
\begin{align}
  \EE[B_\tau] = {1\over 2} \EE[\tau] = {1\over 2} \EE \left[\sum_{k=1}^\infty X_k^2\right]\,.
\end{align}
Consequently, we have
\begin{align} I(W;Y^\infty) &\le D(P_{Y^\infty | W} \| \Phi | P_W) \\
			 & = \EE\left[\log {dP_{Y^\infty|W}(Y^\infty|W)\over d\Phi(Y^\infty)}\right] \\
			 & = {2 \log e\over N_0} \EE[B_\tau]\\
			 & \le {E\over N_0} \log e\,, \label{eq:laststep}
\end{align}			 
where \eqref{eq:laststep} follows from the energy constraint~\eqref{E:Eav}.
\end{proof}


Our next achievability result studies the performance of a variable-length separated scheme. 
\begin{thm}
 Fix positive $E_1$ and $E_2$ such that
\begin{equation}
E_1 + E_2 \leq E. \label{E:twoE}
\end{equation}
Denote
\begin{align}
&~ \varepsilon(E, m) \\
\triangleq&~  1 -  \frac 1 {\sqrt {\pi N_0}}  \int_{-\infty}^\infty  \left( 1 - Q\left( \frac{x + \sqrt E } {\sqrt{\frac{N_0}{2}}}\right) \right)^{m-1} \mathrm{e}^{ - \frac{x^2}{N_0}}\, \mathrm{d}x.  \notag
\end{align}
Assume that $W$ takes values in $\{1, 2, \ldots, M\}$. 
 There exists an $(E, \epsilon)$ non-feedback code for the transmission of random variable $W$ over an AWGN channel without feedback  such that
\begin{align}
\epsilon &\leq   \E{\varepsilon \left(E_1, W \right) } + \varepsilon(E_2, \lfloor \log_2 M \rfloor + 1) \label{E:A}.
\end{align}
\label{E:thm:A}
\end{thm}

\begin{proof}
Assume that the outcomes of $W$ are ordered in decreasing probabilities. Consider the following variable-length separated achievability scheme: the source outcome $m$ is first losslessly represented as a binary string of length $\lfloor \log_2 m \rfloor$ by assigning it to the $m$-th binary string in $\{\varnothing, 0, 1, 00, 01, \ldots\}$ (the most likely outcome is represented by the empty string). Then, all binary strings are grouped according to their encoded lengths. A channel codebook is generated for each group of sequences. The encoded length is sent over the channel with high reliability, so the decoder almost never makes an error in determining that length. Then the encoder makes an ML decision only between sequences of that length. A formal description and an error analysis follow.

{\it Codebook:} 
 the collection of $M + \lfloor \log M \rfloor + 1$ codewords
\begin{align}
 \mathbf c_j &= \sqrt {E_1}\, \mathbf e_j, \ j = 1, 2, \ldots, M \\
 \mathbf c_j &= \sqrt{E_2}\, \mathbf e_j, \ j = M + 1, \ldots, M + \lfloor \log_2 M \rfloor + 1\label{ed:codebook}
\end{align}
 where $\{\mathbf e_j, \ j = 1, 2, \ldots\}$ is an orthonormal basis of $L_2(\mathbb R^\infty)$.

{\it Encoder:} The encoder sends the pair $( m, \lfloor \log_2 m \rfloor)$ by transmitting $\mathbf c_m + \mathbf c_{M + \lfloor \log_2 m \rfloor + 1}$. 

{\it Decoder:} Having received the infinite string corrupted by i.i.d. Gaussian noise $\mathbf z$, the decoder first (reliably) decides between $\lfloor \log_2 M \rfloor + 1$ possible values of $\lfloor \log_2 m \rfloor$ based on the minimum distance: 
\begin{equation}
\hat \ell \triangleq \argmin_j \| \mathbf z - \mathbf c_{M + j + 1} \|, \ j =  0, \ldots, \lfloor \log_2 M \rfloor
\end{equation}
As shown in \cite[p. 258]{wozencraft1965principles}),\cite{pyatoshin1968some}, \cite[Theorem 3]{polyanskiy2011minimumenergy}, the probability of error of such a decision is given by $\varepsilon(E, \lfloor \log_2 M \rfloor + 1)$. This accounts for the second term in \eqref{E:A}. 
The decoder then decides between $2^{\hat \ell}$ messages\footnote{More precisely, $2^{\hat \ell}$ messages if $\hat \ell \leq \lfloor \log_2 M \rfloor - 1$ and $M - 2^{\lfloor \log_2 M \rfloor} + 1 \leq 2^{\lfloor \log_2 M \rfloor}$ messages if $\hat \ell = \lfloor \log_2 M \rfloor$. } $j$ with $\lfloor \log j \rfloor = \hat \ell$:
\begin{equation}
\hat {\mathbf c} \triangleq \argmin \| \mathbf z - \mathbf c_j \|, \ j =  2^{\hat \ell}, \ldots,  \min\{ 2^{\hat \ell + 1}  - 1,  M \}
\end{equation}
The probability of error of this decision rule is similarly upper bounded by $\varepsilon \left(E, m \right)$, provided that the value of $\lfloor \log_2 m \rfloor$ was decoded correctly: $\hat \ell = \lfloor \log_2 m \rfloor$. Since $2^{ \lfloor \log_2 m \rfloor } \leq m$, this accounts for the first term in \eqref{E:A}.

\end{proof}

\apxonly{
Probability of error in random coding with MAP decoder, given $S =1 $ is transmitted, satisfies
\begin{equation}
P_{e| S = 1} = 1 - \frac 1 {\sqrt {\pi N_0}}\int_{-\infty}^\infty \prod_{m \neq 1} \left( 1 - Q\left( \frac{z + \sqrt E + \frac{ 1} {\sqrt E} \log \frac {P_{S}(1)} {P_S(m)}  } {\sqrt{\frac{N_0}{2}}}\right) \right) \mathrm{e}^{ - \frac{z^2}{N_0}} dz
\end{equation}
}

Normally, one would choose $1 \ll E_2 \ll E_1$ so that the second term in \eqref{E:A}, which corresponds to the probability of decoding the length incorrectly, is negligible compared to the first term, and the total energy $E \approx E_1$. Moreover, if $W$ takes values in a countably infinite alphabet, one can truncate it
so that the tail is negligible with respect to the first term in \eqref{E:A}. To ease the evaluation of the first term in \eqref{E:A}, one might use $i \leq \frac 1 {P_W(i)}$. In the equiprobable case, this weakening leads to $\E{\varepsilon \left(E_1, W \right) }  \leq \varepsilon \left(E_1, M \right)$.

If the power constraint is average rather than maximal, a straightforward extension of Theorem \ref{E:thm:A} ensures the existence of an $(E, \epsilon)$ code (average power constraint) for the AWGN channel with 
\begin{align}
\epsilon &\leq   \E{\varepsilon \left(E_1( \lfloor \log_2 W \rfloor ), W \right) } + \varepsilon(E_2, \lfloor \log M \rfloor + 1),\label{E:Aa}
\end{align}
where $E_1 \colon \{0, 1, \ldots, \lfloor \log_2 M \rfloor \} \mapsto \mathbb R_+$ and $E_2 \in \mathbb R_+$ are such that
\begin{equation}
\E{ E_1( \lfloor \log_2 W \rfloor ) } + E_2 \leq E. 
\end{equation}

\section{Asymptotic expansions of the energy-distortion tradeoff}
\label{sec:ed}
\subsection{Problem setup}

This section focuses on the energy-distortion tradeoff in the JSCC problem. Like in Section \ref{sec:Ecodes}, we limit the total available transmitter energy $E$ without any restriction on the
(average) number of channel uses per source sample. Unlike Section \ref{sec:Ecodes}, we allow general (not neccesarily discrete) sources, and we study the
tradeoff between the source dimension $k$, the total energy $E$ and the fidelity of reproduction. Thus, we identify the
minimum energy compatible with a given target distortion without any restriction on the time-bandwidth product (number of
degrees of freedom).  As in Section \ref{sec:def}, we consider both the average and the excess distortion criteria.

Formally, we let the source be a $k$-dimensional vector $S^k \in \mathcal S^k$.  A $(k, E, d, \epsilon)$ energy-limited code is an energy-limited code for $S^k$ with total energy $E$ and probability $ \leq \epsilon$ of distortion exceeding $d$ (see \eqref{eq:def1}). Similarly,  a $(k, E, d)$ energy-limited code is an energy-limited code for $S^k$ with total energy $E$ and average distortion not exceeding $d$ (see \eqref{eq:def2}). In the remainder of this section, we characterize the minimum energy required to transmit $k$ source samples at a given fidelity, i.e. to characterize the following fundamental limits: 
\begin{align}
E^\star_f (k, d) &\triangleq  \left\{ \inf E \colon \exists \text{ a } (k, E, d) \text{ feedback code}\right\},\\
E^\star_f (k, d, \epsilon) &\triangleq  \left\{ \inf E \colon \exists \text{ a } (k, E, d, \epsilon) \text{ feedback code}\right\}
\end{align}
as well as the corresponding limits $E^\star(k, d)$ and $E^\star (k, d, \epsilon)$ of the energy-limited non-feedback codes.

\subsection{Previous results on the energy-per-bit and the energy-distortion tradeoff}
If the source produces equiprobable binary strings of length $k$,  Shannon \cite{shannon1949communication} showed that the minimum energy per information bit to noise power spectral density
ratio compatible with vanishing block error probability converges to
\begin{equation}
\ \frac{E^\star(k, 0, \epsilon)}{k N_0} \to \log_e 2 = - 1.59 \text{ dB} \label{ed:shannon}
\end{equation}
as $k \to \infty$, $\epsilon \to 0$. The fundamental limit in \eqref{ed:shannon} holds regardless of whether feedback is available. Moreover, this fundamental limit is known to be the same regardless of whether the channel is subject to fading or whether the receiver is coherent or not \cite{verdu2002spectral}.   
Polyanskiy et al. refined \eqref{ed:shannon} as \cite[Theorem 3]{polyanskiy2011minimumenergy}
\begin{equation}
E^\star(k, 0, \epsilon) \frac{\log e}{N_0} = k + \sqrt{2 k \log e  }\, Q^{-1}(\epsilon) - \frac 1 2 \log k + \bigo{1} \label{E:2orderchannel}
\end{equation}
for transmission without feedback, and as \cite[Theorem 8]{polyanskiy2011minimumenergy}
\begin{equation}
 E^\star_f(k, 0, \epsilon) \frac{\log e}{N_0} = ( 1 - \epsilon) k + \bigo{ 1 } \label{E:2orderfchannel}
\end{equation}
for transmission with feedback. Moreover,  \cite[Theorem 7]{polyanskiy2011minimumenergy} (see also \eqref{E:pol0}) shows that in fact
\begin{equation}
   E^\star(k, 0, 0) \frac{\log e}{N_0} = k,   
\end{equation}
i.e. in the presence of full noiseless feedback, Shannon's limit \eqref{ed:shannon} can be achieved with equality already at $k = 1$ and $\epsilon = 0$. 

For  the finite blocklength behavior of energy per bit in fading channels, see~\cite{yang2016minimum}.

For the transmission of a memoryless source over the AWGN channel under an average distortion criterion, Jain et al. \cite[Theorem 1]{jain2012energy} pointed out that as $k \to \infty$, 
\begin{equation}
 \ \frac{  E^\star(k, d)}{k} \frac{\log e}{N_0} \to  R(d)  \label{ed:jain}.
\end{equation}
Note that \eqref{ed:jain} still holds even if noiseless feedback is available.

Unlike Polyanskiy et al. \cite{polyanskiy2011minimumenergy}, we allow analog sources and arbitrary distortion criteria, and unlike Jain et al. \cite{jain2012energy}, we are interested in a nonasymptotic analysis of the minimum energy per sample.

\subsection{Energy-limited feedback codes}
\label{sec:edav_fb}
Our first result in this section is a refinement of \eqref{ed:jain}. 

\begin{thm}\label{th:edav} Let the source and its distortion measure satisfy assumptions A\ref{item:first}--A\ref{item:last}. The minimum energy required to transmit $k$ source symbols with average distortion $\leq d$ over an AWGN channel with feedback  satisfies
\begin{equation}
 E^\star_f(k, d) \cdot {\log e\over N_0} = k R(d) + O(\log k) \label{eq:edav}.
\end{equation}
\end{thm}

\begin{proof}
{\it Achievability}. The expansion in \eqref{eq:edav} is achieved by the following separated source/channel scheme. For the source code, we use the code of Yang and Zhang \cite{yang1999redundancy} (abstract alphabet) that compresses the source down to $M$ representation points with average distortion $d$ such that
\begin{equation}
 \log M = kR(d) + O(\log
k).
\end{equation}
For the channel code, we transmit the binary representation of $M$ error-free using the optimal scheme of Polyanskiy et al. \cite[Theorem 7]{polyanskiy2011minimumenergy}, so that 
\begin{equation}
\log M =  \frac{E}{N_0} \log e.  
\end{equation}
 
 {\it Converse}. By data processing, similarly to \eqref{E:Cf}, 
\begin{equation}
k R(d) \leq \frac {E}{N_0} \log e. 
\end{equation}
\end{proof}

\begin{remark}
 For the transmission of a Gaussian source over the feedback AWGN channel, we have
 \begin{equation}
 E^\star_f(k, d) \cdot {\log e\over N_0} = k R(d) \label{eq:edavexact}.
\end{equation}
Indeed, the Schalkwijk-Bluestein   scheme \cite{schalkwijk1967transmission,cruise1967achievement} attains \eqref{eq:edavexact} for $k = 1$. For $k > 1$, transmitting the Schalkwijk-Bluestein codewords corresponding to $i$-th source sample 
in time slots $i, k + i, 2 k + i, \ldots$ attains \eqref{eq:edavexact} exactly for all $k = 1, 2, \ldots$. 
\end{remark}

\begin{thm}
In the transmission of a source satisfying the assumptions A\ref{item:first}--A\ref{item:last} over an AWGN channel with feedback, the minimum average energy required for the transmission of $k$ source samples under the requirement that the probability of exceeding distortion $d$ is no greater than $0 \leq \epsilon < 1$ satisfies, as $k \to \infty$, 
\begin{align}
 E^\star_f \left( k, d, \epsilon \right) \frac{\log e}{N_0}   &= (1 - \epsilon) k R(d) - \sqrt{\frac{k \mathcal V(d)}{2 \pi} } \mathrm{e}^{- \frac { (\Qinv{\epsilon})^2} 2 }  \notag\\
  &+ \bigo{\log k}. \label{E:edex}
\end{align}
\label{th:edex}
\end{thm}

\begin{proof}
{\it Achievability.} Pair a lossy compressor $S^k \to W$ with excess-distortion probability $\epsilon$ and $H(W) = H_{d, \epsilon}(S^k)$ with the achievability scheme in Theorem \ref{thm:ECfeedback} and apply \eqref{E:Af} and \eqref{minI2orderlossy}. 

{\it Converse.} Again, the converse result follows proceeding as in \eqref{E:Cf}, invoking \eqref{minI2orderlossy}. 
\end{proof}

Comparing \eqref{E:edex} and \eqref{f2v:2order}, we observe that, similar to the setup in Section \ref{sec:main}, allowing feedback and average power constraint reduces the asymptotically achievable minimum energy per sample by a factor of $1 - \epsilon$. As in Section \ref{sec:main}, that limit is approached from below rather than from above, i.e. shorter blocklenghts are more economical in terms of the energy
per transmitted symbol.

Similar to the setup of Section \ref{sec:main}, naive separation achieves at most
\begin{equation}
 E_f^\star(k, d, \epsilon) \frac {\log e}{N_0} \leq (1 - \epsilon) k R(d) + a \sqrt{k \log k}, \, a > 0.
\end{equation}

\subsection{Energy-limited non-feedback codes}
\label{sec:edav_nonfb}

Our next result generalizes \cite[Theorem 3]{polyanskiy2011minimumenergy}. Loosely speaking, it shows that the
energy $E$, probability of error $\epsilon$ and distortion $d$ of the best non-feedback code satisfy
\begin{equation}
  {E\over N_0} \log e - k R(d) \approx \sqrt{k\mathcal{V}(d) + {2E\over N_0}\log e} \cdot Q^{-1}(\epsilon). \label{eq:116a}
\end{equation}
Note that in \eqref{eq:116a} source and channel dispersions add up, as in the usual (non-feedback) joint source-channel coding
problem~\cite{wang2011dispersion,kostina2012jscc}.
More precisely, we have the following:
\begin{thm}
In the transmission of $k$ samples of a stationary memoryless source (satisfying the assumptions A\ref{item:first}--A\ref{item:last}) over the AWGN channel, the minimum energy necessary for achieving probability $0< \epsilon<1$ of exceeding distortion $d$ satisfies, as $k \to \infty$, 
\begin{align}
&~  E^\star \left( k, d, \epsilon \right) \frac{\log e}{N_0}  \label{E:2order}\\
=&~ k R(d)  + \sqrt{ k \left( 2 R(d) \log e + \mathcal V(d)\right) } \Qinv{\epsilon} + \bigo{ \log k }. \notag
\end{align} 
\label{thm:E2order}
\end{thm}

\begin{proof}

{\it Achievability:} 
We let the total energy $E$ be such that 
\begin{align}
 E\, \frac{\log e}{N_0}  
&= k R(d) + \sqrt{ k \left( 2 R(d) \log e + \mathcal V(d)\right) } \Qinv{\epsilon - \frac a {\sqrt
 k}} \notag\\
 &+ b \log k, \label{E:2order_ach}
\end{align}
and we show that $a > 0$ and $b$ can be chosen so that the excess distortion probability is bounded by $\epsilon$. 

We consider a good lossy code with $M = \exp(2 k R(d))$ representation points, so that the probability that the source
is not represented within distortion $d$ is exponentially small. We show that a combination of that code with the
variable-length separated scheme in Theorem \ref{E:thm:A} achieves \eqref{E:2order_ach}. 
First, we prove the following generalization of Theorem  \ref{E:thm:A} to the lossy case: for any $M$, there exists an $(k, E, d, \epsilon^\prime)$ code for the AWGN channel (without feedback)  such that
\begin{align}
\epsilon^\prime &\leq   \E{\varepsilon \left(E_1, \frac 1 {P_{\mathsf Z^\star}^k (B_d(S^k))} \right) } + \varepsilon(E_2, \lfloor \log M \rfloor + 1) \notag\\
&+ \E{\left( 1 - P_{\mathsf Z^\star}^k (B_d(S^k)))\right)^M }\label{E:Alossy},
\end{align}
where
\begin{equation}
 B_d(s^k) \triangleq \{z^k \in \hat {\mathcal S}^k \colon \mathsf d(s^k, z^k) \leq d\}.
\end{equation}
Towards that end, let the representation points $Z_1, Z_2, \ldots, Z_M$ be drawn i.i.d. from $P_{\mathsf Z^\star}^k$. The source encoder goes down the list of the representation points and outputs the index of the first $d$-close match to $S^k$:
\begin{equation}
 W \triangleq \min \{ i \colon \mathsf d(S^k, Z_i) \leq d \}
\end{equation}
(if there is no such index, it outputs 1). Averaged over $Z_1, \ldots, Z_M$, the probability that no $d$-close match is found is upper bounded by the third term in \eqref{E:Alossy} (e.g. \cite[Theorem 10]{kostina2011fixed}). The index $W$  
is then transmitted over the channel using the scheme in Theorem  \ref{E:thm:A}, with the total probability of error averaged over all lossy codebooks given by
\begin{align}
 \epsilon^\prime &\leq   \E{\varepsilon \left(E_1, W \right) } + \varepsilon(E_2, \lfloor \log M \rfloor + 1) \notag\\
 &+ \E{\left( 1 - P_{\mathsf Z^\star}^k (B_d(S^k)))\right)^M }. 
\end{align}
Since conditioned on $S = s$, $W$ is geometrically distributed with success probability $P_{Z^\star}^k(B_d(s^k))$, we have 
\begin{equation}
 \E{W | S^k = s^k} = \frac 1 {P_{\mathsf Z^\star}^k(B_d(s^k))}.
\end{equation}
Noting that $\varepsilon(E, m)$ is a concave function of $m$, we have by Jensen's inequality
\begin{equation}
 \E{\varepsilon \left(E,W \right)} \leq \E{\varepsilon \left(E, \E{W|S^k} \right)\} },
\end{equation}
which gives the first term in \eqref{E:A}. 

We proceed to show that with the choice of
\begin{equation}
E_1 =  E - c\log E, \label{E:Echoice}
\end{equation}
for an appropriate $c > 0$, and $M = \exp(2 k R(d))$, the right side of \eqref{E:Alossy} is upper bounded by $\epsilon$. 

A reasoning similar to \cite[(108)--(111)]{kostina2011fixed} and the Cram\'er-Chernoff bound yield
\begin{align}
 \E{\left( 1 - P_{Z^{k \star}}(B_d(S^k))\right)^M } \leq \exp(- k a_1) \label{E:scerror}
\end{align}
for some $a_1 > 0$. On the other hand, \eqref{E:2orderchannel} \cite[Theorem 3]{polyanskiy2011minimumenergy} implies  
\begin{equation}
\varepsilon \left( E, m \right)  = \Prob{ \log m  >  G(E) }+ \bigo{\frac 1 {\sqrt k}}  \label{E:cherror},
\end{equation}
where 
$G(E)  \stackrel{D}{=} \mathcal N \left( \frac {E }  {N_0} \log e - \frac 1 2 \log \frac E {N_0}, \frac {2 E }{N_0} \log^2 e \right)$.

Applying \eqref{E:scerror} and  \eqref{E:cherror} to \eqref{E:A}, we conclude that the excess-distortion probability is bounded above by
\begin{align}
&~\Prob{ \log \frac 1 {P_{Z^{k \star}}(B_d(S^k))}   >  G(E - c \log E) }  \notag\\
+&~ \Prob{\log \left( \log M + 1 \right)  >  G\left( c \log E\right) } +  \bigo{\frac 1 {\sqrt E}} \label{E:2error}.
\end{align}
The second term in \eqref{E:2error} can be made to decay as fast as $\bigo{\frac 1 {\sqrt E}}$ for large enough $c$. To evaluate the first term in \eqref{E:2error}, we recall \cite[Lemma 2]{kostina2011fixed}, which states that for $k$ large enough,
\begin{align}
 &~ \Prob{ \log \frac 1 {P_{Z^{k\star}}(B_d(S^k))} \leq \sum_{i = 1}^k \jmath_{\mathsf S}(S_i, d) + C \log k + c_0 } \notag\\
 \geq& ~ 1 - \frac K {\sqrt k} \label{E:kvlemma},
\end{align}
where $c_0$ and $K$ are constants, and $C$ is a constant explicitly identified in \cite[Lemma 2]{kostina2011fixed}. Letting $b = c + C - \frac 1 2$ and applying \eqref{E:kvlemma} to upper-bound the first term in  \eqref{E:2error}, we conclude by the Berry-Ess\'een Theorem that
\begin{align}
&~ \Prob{ \log \frac 1 {P_{Z^{k \star}}(B_d(S^k))}    >  G\left(E - c \log E\right) } \notag\\
\leq&~ \epsilon - \frac a {\sqrt k} + \bigo{\frac 1 {\sqrt k}}.
\end{align}
Since $a$ can be chosen so that \eqref{E:2error} is upper bounded by $\epsilon$ for $k$ large enough, this concludes the proof of the achievability part.

{\it Converse:} 
The result in \cite[Theorem 1]{kostina2012jscc} states that the excess distortion probability is bounded as in 
\setcounter{mytempeqncnt}{\value{equation}}
\begin{table*}[!b]
\normalsize
\setcounter{equation}{\value{mytempeqncnt}}
\vspace*{4pt}
\hrulefill
\begin{align}
 \epsilon
\geq &~\sup_{\gamma > 0}\bigg\{ \sup_{P_{\bar {\mathbf Y}}} \E{\inf_{\mathbf x \colon \| \mathbf x\|^2 \leq E}  \Prob{ \jmath_{S^k}(S^k, d) -  \imath_{\mathbf X; \bar{\mathbf Y}}(\mathbf x; \mathbf Y) \geq \gamma \mid  S} } 
 - \exp\left( -\gamma \right) \bigg\} \label{eq:C11}
\end{align}
\end{table*}
\eqref{eq:C11} 
in the bottom of the page, 
where $\jmath_{S^k}(s^k, d)$ is the $\mathsf d$-tilted information in $s^k \in \mathcal S^k$ defined in \eqref{jE}, 
$\imath_{\mathbf X; \bar{\mathbf Y}} (\mathbf x; \mathbf y) \triangleq \log \frac{dP_{\mathbf Y | \mathbf X = \mathbf x}}{dP_{\bar {\mathbf Y}}}(\mathbf y)$, and 
$P_{\mathbf Y | \mathbf X = \mathbf x}$ and $P_{\bar {\mathbf Y}}$ are specialized to
\begin{align}
P_{\bar {\mathbf Y}} &\stackrel{D}{=} \prod_{k = 1}^\infty \mathcal N \left( 0, \frac {N_0}{2}\right) \\
P_{\mathbf Y | \mathbf X = \mathbf x} &\stackrel{D}{=} \prod_{k = 1}^\infty \mathcal N \left(x_k, \frac {N_0}{2} \right)\\
\imath_{\mathbf X; \bar{\mathbf Y}} (\mathbf x; \mathbf Y)  &\stackrel{D}{=} 
\mathcal N \left( \frac{\|\mathbf x\|^2}{N_0} \log e, \frac{2\|\mathbf x\|^2}{N_0} \log e\right) 
\end{align}

Next, we let in \eqref{eq:C11} $\gamma = \frac 1 2 \log \frac {E}{N_0}$. Since $\jmath_{S^k}(S^k, d) = \sum_{i = 1}^k \jmath_{\mathsf S}(S_i, d)$ is a sum of independent random variables, the Berry-Ess\'een bound applies to the probability in \eqref{eq:C11}, and the converse direction of \eqref{E:2order} follows since $\| \mathbf x\|^2 \leq E$.
\apxonly{
\textbf{TODO:} Vic, is it possible to give a short proof of~\eqref{eq:C11} with just $\beta_\alpha$ and the choice of
$Q$-measure? At least in the apxonly version? VK: Yur I am not sure what you mean. The proof of \eqref{eq:C11} in \cite{kostina2012jscc} is already short?
}
\end{proof}

If the maximal power constraint in \eqref{E:Emax} is relaxed to \eqref{E:Eav},  then $E^\star_a \left( k, d, \epsilon \right)$, the minimum average power required for transmitting $k$ source samples over an AWGN channel with the probability of exceeding distortion $d$ smaller than or equal to $0 < \epsilon < 1$ satisfies, under assumptions A\ref{item:first}--A\ref{item:last}: 
\begin{equation}
 E^\star_a \left( k, d, \epsilon \right) \frac{\log e}{N_0}   = R(d) (1 - \epsilon) k  + \bigo{\sqrt{k \log k}}, \label{E:2ordera}
\end{equation}
i.e. the asymptotically achievable minimum energy per sample is reduced by a factor of $1 - \epsilon$ if a maximal power constraint is relaxed to an average one. This parallels the result in \eqref{f2v:avnof}, which shows that variable-length coding over a channel reduces the asymptotic fundamental limit by a factor of $1 - \epsilon$ compared to fixed-length coding, even without feedback.
\begin{proof}[Proof of \eqref{E:2ordera}]
Observe that Theorem \ref{th:edex} ensures that a smaller average energy than that in \eqref{E:2ordera} is not attainable even with full noiseless feedback. 
In the achievability direction, let $\left( \mathsf f^\star, \mathsf g^\star \right)$ be the optimal variable-length source code achieving the probability of exceeding $d$ equal to $\epsilon^\prime$ (see \cite[Section III.B]{kostina2015varrate}). Denote by $\ell(\mathsf f^\star(s))$ the length of $\mathsf f^\star(s)$. Let $M$ be the size of that code. Set the energy to transmit the codeword of length $\ell(\mathsf f^\star(S^k))$ to
\begin{equation}
 \ell(\mathsf f^\star(S^k)) N_0 \log_e 2 + \sqrt{k \log k} \label{E:etcl}.
\end{equation}
As shown in \cite{kostina2015varrate}, $\E{\ell(\mathsf f^\star(S))}$ is equal to the right side of \eqref{f2v:2order} (with $\epsilon$ replaced by $\epsilon^\prime$). Choosing $\epsilon^\prime = \epsilon - \frac a {\sqrt k}$ for some $a$, we conclude that indeed the average energy satisfies \eqref{E:2ordera}. Moreover, \eqref{E:cherror} implies that the expression inside the expectation in \eqref{E:Aa} is $\bigo{\frac 1 {\sqrt k}}$. It follows that for a large enough $a$, the excess distortion probability is bounded by $\epsilon$. 
\end{proof}

%
%
%
%
%
%

\apxonly{
TODO:
Average energy - average distortion tradeoff: 
\begin{enumerate}
\item Observe: for sources with separable distortion measure
$d^\star(k + 1, E) \leq d^\star(k, E)$ (time sharing argument, average distortion). 
\item Algorithm to send 1 Gaussian sample at the optimal tradeoff.
\end{enumerate}
}

\section{Conclusions}

We have considered several scenarios for joint source-channel coding with and without feedback. Our main conclusions are: 

\begin{enumerate}

\item The average delay vs. distortion tradeoff with feedback, as well as the average energy vs. distortion tradeoff with feedback, is governed by the channel capacity, and the source rate-distortion and rate-dispersion functions. In particular, the channel dispersion plays no role.

\item In variable-length coding with feedback, the asymptotically achievable minimum average length is reduced by a factor of $1 - \epsilon$, where $\epsilon$ is the excess distortion probability. This asymptotic fundamental limit is approached from below, i.e., counter-intuitively, smaller source blocklengths lead to smaller attainable average delays.

\item Introducing a termination symbol that is always decoded error-free allows for transmission over noisy channels with guaranteed distortion. 

\item Variable-length transmission without feedback still improves the asymptotic fundamental limit by a factor of $1 - \epsilon$, where $\epsilon$ is the excess distortion probability.

\item In all the cases we have analyzed the approach to the fundamental limits is very fast: $\bigo{\frac {\log k}{k}}$, where $k$ is the source blocklength.  This behavior is attained, under average distortion,  by a separated scheme with stop-feedback.

\item The setting of a wideband Gaussian channel with an energy constraint exhibits many interesting parallels with the variable-length coding setting. Allowing a non-vanishing excess distortion probability $\epsilon$ and an average (rather than maximal) energy constraint reduces the asymptotically achievable minimum energy by a factor of $1 - \epsilon$. In the presence of feedback, just as in the variable-length coding, this asymptotic fundamental limit is approached from below.

\item  Error-free transmission of a random variable $W$ over the AWGN channel with ideal feedback, with almost optimal energy consumption, is possible via a variable-length separated scheme. 

\item More generally, variable-length separated schemes perform remarkably well in all considered scenarios. 
\end{enumerate}

\appendices

\bibliographystyle{IEEEtran}
\bibliography{IEEEabrv,rateDistortion,reports}

\begin{IEEEbiographynophoto}{Victoria Kostina}(S'12--M'14)
joined Caltech as an Assistant Professor of Electrical Engineering in the fall of 2014. She holds a Bachelor's degree from Moscow institute of Physics and Technology (2004), where she was affiliated with the Institute for Information Transmission Problems of the Russian Academy of Sciences, a Master's degree from University of Ottawa (2006), and a PhD from Princeton University (2013). Her PhD dissertation on information-theoretic limits of lossy data compression received Princeton Electrical Engineering Best Dissertation award.  She is also a recipient of Simons-Berkeley research fellowship (2015). Victoria Kostina's research spans information theory, coding, wireless communications and control. 
\end{IEEEbiographynophoto}

\begin{IEEEbiographynophoto}{Yury Polyanskiy}(S'08--M'10--SM'14) is an 
Associate Professor of Electrical Engineering and Computer Science and a member of LIDS at MIT.
Yury received M.S. degree in applied mathematics and physics from the Moscow Institute of Physics and Technology,
Moscow, Russia in 2005 and Ph.D. degree in electrical engineering from Princeton
University, Princeton, NJ in 2010. Currently, his research focuses on basic questions in information theory, error-correcting codes, wireless communication and fault-tolerant and defect-tolerant circuits.
Yury won the 2013 NSF CAREER award and 2011 IEEE Information Theory Society Paper Award.
\end{IEEEbiographynophoto}

\begin{IEEEbiographynophoto}{Sergio Verd\'{u}}(S'80--M'84--SM'88--F'93)
 received the Telecommunications Engineering degree from the
Universitat Polit\`{e}cnica de Barcelona in 1980, and the Ph.D. degree in Electrical Engineering from the
University of Illinois at Urbana-Champaign in 1984. Since then, he has been a member of the faculty of
Princeton University, where he is the Eugene Higgins Professor of Electrical Engineering, and is a member
of the Program in Applied and Computational Mathematics.

Sergio Verd\'{u} is  the recipient of 
the 2007 Claude E. Shannon Award, and
the 2008 IEEE Richard W. Hamming Medal. 
He is a member of both the National Academy of Engineering and the National Academy of Sciences. 
In 2016, Verd\'{u} received the National Academy of Sciences Award for Scientific Reviewing.

Verd\'{u} is a recipient of several paper awards from the IEEE: 
the 1992 Donald Fink Paper Award, 
the 1998 and 2012 Information Theory  Paper Awards, 
an Information Theory Golden Jubilee Paper Award,
the 2002 Leonard Abraham Prize Award,  
the 2006 Joint Communications/Information Theory Paper Award, 
and the 2009 Stephen O. Rice Prize from the IEEE Communications Society.  
In 1998, Cambridge University Press published his book {\em Multiuser Detection,} 
for which he received the 2000 Frederick E. Terman Award from the American Society for Engineering Education. 
He was awarded a Doctorate Honoris Causa from the Universitat  Polit\`{e}cnica de Catalunya in 2005, and was elected corresponding member of the Real Academia de Ingenier\'{i}a of Spain in 2013.

Sergio Verd\'{u} served as President of the IEEE Information Theory Society in 1997, and
on its Board of Governors (1988-1999, 2009-2014).
He has also served in various editorial capacities for the {\em IEEE Transactions on Information Theory}:
Associate Editor (Shannon Theory, 1990-1993; Book Reviews, 2002-2006),  
Guest Editor of the Special Fiftieth Anniversary Commemorative Issue
(published by IEEE Press as ``Information Theory: Fifty years of discovery"), 
and member of the Executive Editorial Board (2010-2013).
He co-chaired the  Europe-United States {\em Frontiers of Engineering} program, of
the National Academy of Engineering during 2009-2013.
He is the founding Editor-in-Chief of {\em Foundations and Trends in Communications and Information Theory}. 
Verd\'{u} served as co-chair of the 2000 and 2016 {\em IEEE International Symposium on Information Theory}.

Sergio Verd\'{u} has held visiting appointments at 
the Australian National University, 
the Technion-Israel Institute of Technology, 
the University of Tokyo, 
the University of California, Berkeley, 
the Mathematical Sciences Research Institute, Berkeley, 
Stanford University, 
and the Massachusetts Institute of Technology.
\end{IEEEbiographynophoto}

\end{document}